\documentclass[twocolumn]{autart}

\usepackage{amsmath}
\usepackage{amsfonts} 
\usepackage{amssymb}
\usepackage{upgreek}
\usepackage{graphicx}
\usepackage{float}

\def\usetikz{0} 
\if\usetikz1
	\usepackage{tikz}
	\usepackage[nooldvoltagedirection]{circuitikz} 
	\usepackage{pgfplots}
	\pgfplotsset{compat=1.14}
	\usepgfplotslibrary{groupplots}
	\usetikzlibrary{positioning}
	\tikzset{terminal/.style=[node distance=1.25cm,
		fblock/.style={rectangle,minimum size=1.0cm,
		thick,draw=black},
		sum/.style={circle,minimum size=0.5cm,
		thick,draw=black},
		every new ->/.style={-latex},
		every new --/.style={},
		hvpath/.style={to path={-| (\tikztotarget)}},
		vhpath/.style={to path={|- (\tikztotarget)}}
	}
	\usetikzlibrary{external} 
\else
	\usepackage{color}
	\graphicspath{{./tikz/}}
\fi

\theoremstyle{definition}
\newtheorem{cond}{Condition}

\newcommand{\Iapp}{i_{\text{app}}}
\newcommand{\Iint}{g}
\newcommand{\Inoi}{e}
\newcommand{\w}{w}

\newcommand{\vref}{r}
\newcommand{\samplingT}{t_s}
\newcommand{\inpbound}{\beta}
\newcommand{\gain}{\gamma}
\newcommand{\gmax}{\bar{g}}
\newcommand{\npars}{{n_\theta}}
\newcommand{\nstates}{{n_x}}
\newcommand{\ngat}{{n_\w}}
\newcommand{\ninputs}{{n_{u}}}
\newcommand{\noutputs}{{n_{y}}}
\newcommand{\nref}{{n_{\vref}}}
\newcommand{\nchannels}{{n_{\text{c}}}}
\newcommand{\nmodel}{{n_{\text{m}}}}

\newcommand{\setreal}{\mathbb{R}}
\newcommand{\setint}{\mathbb{Z}}
\newcommand{\setnat}{\mathbb{N}}

\newcommand{\inpclass}{\mathcal{U}_\inpbound}

\newcommand{\pardomain}{\mathcal{D}}

\begin{document}

\begin{frontmatter}

\title{Feedback Identification of conductance-based models\thanksref{footnoteinfo}} 

\thanks[footnoteinfo]{This paper was not presented at 
any IFAC meeting. Corresponding author T.~B.~Burghi. }

\author[cambridge]{Thiago B. Burghi}\ead{tbb29@cam.ac.uk},    	
\author[eindhoven]{Maarten Schoukens}\ead{m.schoukens@tue.nl},  
\author[cambridge]{Rodolphe Sepulchre}
					\ead{r.sepulchre@eng.cam.ac.uk} 

\address[cambridge]{Department of Engineering, Control Group, 
University of Cambridge, Cambridge CB2 1PZ, UK.}                                             
\address[eindhoven]{Department of Electrical Engineering,  
Eindhoven University of Technology, 5612 AZ Eindhoven,
Netherlands.}            

\begin{keyword}      	   
Nonlinear system identification;
Closed-loop identification;
Prediction error methods;
Contraction analysis;
Neuronal models.	  	 
\end{keyword}            

\begin{abstract} 

This paper applies the classical prediction error
method (PEM) to the estimation of nonlinear discrete-time
models of neuronal systems subject to input-additive noise. 
While the nonlinear system exhibits excitability, bifurcations,
and limit-cycle oscillations, we prove consistency of the
parameter estimation procedure under output feedback. 
Hence, this paper provides a rigorous framework 
for the application of conventional nonlinear system
identification methods to discrete-time stochastic neuronal 
systems. The main result exploits the elementary property that
conductance-based models of neurons have an exponentially
contracting inverse dynamics. This property is implied by
the voltage-clamp experiment, which has been the
fundamental modeling experiment of neurons ever since the
pioneering work of Hodgkin and Huxley. 

\end{abstract}

\end{frontmatter}		

\section{Introduction}

The estimation of models for biological neuronal systems 
is a topic that has attracted considerable interest in the 
scientific community over the past decades 
\cite{nogaret_automatic_2016,geit_automated_2008,huys_efficient_2006,lepora_efficient_2012,milescu_maximum_2005}.
However, the asymptotic properties of published estimation
methods are rarely discussed. This is understandable for
models that exhibit highly nonlinear dynamics including
excitable behaviors and limit cycle oscillations.

The goal of this paper is to show that rigorous convergence
results can be established in the most classical framework 
of the prediction error method (PEM) 
\cite{ljung_convergence_1978,ljung_system_1999}.
In nonlinear system identification, the convergence 
and consistency analysis of the PEM depends on the
assumption that the signals are generated within a
process with some form of input-output stability --- 
for instance, a fading memory \cite{boyd_fading_1985}, 
input-output exponential stability
\cite{ljung_convergence_1978,novara_parametric_2011,abdalmoaty_linear_2019},
or mean square convergence of the output to that of a Volterra
series \cite{paduart_identification_2010,schoukens_identification_2017}.
In addition, the analysis is greatly simplified by the assumption
that the true system is affected by output-additive noise only
\cite{schoukens_identification_2017,ljung_perspectives_2010}. 

Neuronal models are nonlinear systems that fail to satisfy the
stability and the output-additive noise assumptions. First,
neuronal systems are primarily subject to \textit{input}-additive
noise. This type of noise models the stochastic fluctuations of
currents traversing the neuronal membrane. For a review of the
modeling of noise in neuronal systems, see
\cite{goldwyn_what_2011,gerstner_neuronal_2014}. Furthermore, the
non-equilibrium nature of neuronal behaviors precludes any
reasonable exponential stability or fading memory assumption.

Previous works have studied the application of the PEM under
these unfavorable conditions. When the noise is input-additive,
the difficulty lies in the intractability of analytically 
computing the optimal one-step-ahead predictor (see 
\cite{schon_sequential_2015} for a discussion). As long as the
data-generating process is input-output exponentially stable,
consistent parameter estimates can be obtained in some cases,
e.g., when predictor models are linear in the past outputs 
\cite{abdalmoaty_linear_2019} or when LTI elements of a 
block-oriented model structure are known
\cite{novara_parametric_2011}. When input-output stability is
not guaranteed, as in the case of oscillatory systems, an
alternative to standard PEM analysis must be found. In 
\cite{casas_prediction_2002}, the authors justify with dynamical
systems theory the application of the PEM to identify the linear
element of a Lure-type system with a limit cycle; the authors
assume ergodicity of the system's signals in order to bypass the 
question of stability. In 
\cite{manchester_identification_2011-1}, the authors develop a
method based on transverse contraction analysis to identify 
oscillatory systems under the assumption that all states of
the model are available; no noise considerations are made.

The main observation underlying the present paper is
that while the assumptions that make PEM analysis tractable are
not verified for conductance-based neuronal models, they hold for
their inverse. In other words, conductance-based models verify
these assumptions under high-gain output feedback. This means
that neuronal systems can be identified with classical techniques
by relying on the direct approach of closed-loop system
identification \cite{forssell_closed-loop_1999}. Using
contraction theory \cite{lohmiller_contraction_1998}, we
rigorously justify the use of the direct approach to consistently
estimate discrete-time neuronal models.

We show that the closed-loop approach to the neuronal
system identification problem is fully consistent with the
classical voltage-clamp experiment of Hodgkin and Huxley
\cite{hodgkin_quantitative_1952}. Voltage-clamp
has remained to date the key experimental methodology to
derive a state-space model of a neuron. 
We show that there is flexibility in designing a
contracting output feedback law beyond the high-gain
implementation of voltage-clamp. As in previous work
dealing with Lure systems \cite{burghi_feedback_2019}, we
advocate that feedback design is an integral element of
neuronal system identification, which makes this an
attractive application of closed-loop system
identification theory.

The paper is organized as follows: in Section 
\ref{sec:preliminaries}, we review a number of classical
tools of nonlinear system identification and analysis. In
Section \ref{sec:models}, we introduce the general class of
conductance-based models and show that they have
a contracting inverse. In Section \ref{sec:identification}, we
detail the identification of the inverse dynamics of 
discrete-time neuronal systems with the PEM and discuss the
plausibility of the required assumptions. In Section
\ref{sec:examples}, we illustrate our results using data from
numerical simulations.

\section{Preliminaries}
\label{sec:preliminaries} 

This section reviews two classical results of system
theory: the convergence properties of the prediction
error method \cite{ljung_convergence_1978}, and the system
property of contraction \cite{lohmiller_contraction_1998}.

We use the following notation: For a discrete-time
variable $x_k$, the signal up to time $k$ is denoted by
$x_{[0,k]} = (x_k,x_{k-1},\dotsc,x_0)$. We write
$\setreal_+ = [0,\infty)$, $\setnat = \{1,2,\dotsc\}$, 
and $\setint_+=~\{0,1,\dotsc\}$.
The number $0$ is treated as a scalar or as a vector, with the
dimension implied by the context in which it is used. The
norm $\|\cdot\|$ denotes the Euclidean norm, and 
$\sigma_{\max}[\;\cdot\;]$ denotes the largest singular 
value of a matrix. 
For arbitrary $\inpbound>0$, the class 
of $\ninputs$-valued sequences $u:\setint_+ \to 
\setreal^\ninputs$ such that 
$\sup_{k\in\setint_+} \max_j |u_{j,k}| < \inpbound$
is denoted by $\inpclass^\ninputs$.

\subsection{Parametric Identification of nonlinear
systems with the Prediction Error Method}
\label{sec:PEM} 

Consider a nonlinear stochastic discrete-time system 
represented by
\begin{equation}
	\label{eq:innovations_form}
	y_k = F_k\left(u_{[0,k]};x_0\right) + e_k
\end{equation}
where $u_k \in \setreal^{\ninputs}$ is the
system's input, $y_k \in \setreal^{\noutputs}$ is the
system's output, $e_k\in \setreal^{\noutputs}$ is a
stochastic process such that $E[e_k\;|\;e_{[0,k-1]}]=0$,
$F_k(\cdot)$ is a sequence of deterministic mappings, and
$x_0$ is an initial state.

Assume that the system \eqref{eq:innovations_form} is in a 
feedback loop with an adaptive feedback element given by
\begin{equation}
	\label{eq:feedback_block}
	u_k = H_k\left(y_{[0,k-1]},u_{[0,k-1]},\vref_k\right)
\end{equation}
where $u_k$ is the feedback element's output, $y_k$ is the
output of \eqref{eq:innovations_form}, and $\vref_k \in 
\setreal^\nref$ is an external signal. 

In the prediction error framework, the system 
\eqref{eq:innovations_form} is identified based on $N$
collected input-output data points, given by the
sequences $y_{[0,N]}$ and $u_{[0,N]}$. For this purpose, a
parametric model is used to obtain a prediction 
$\hat{y}_k$ of the output $y_k$. In this paper, we work
with an \textit{output error} predictor model, which is 
represented by a sequence of operators $\hat{F}_k$ such that
\begin{equation}
	\label{eq:model}
	\hat{y}_k(\theta) = \hat{F}_k
	\left(u_{[0,k]};\theta\right) 
\end{equation}
where $\theta \in \pardomain$ denotes a vector of
parameters, and $\pardomain$ is a subset of 
$\setreal^\npars$, with $\npars$ the number of 
parameters\footnote{In
\eqref{eq:innovations_form} and \eqref{eq:model}, we allow
the input to affect the output without a delay. This
differs from the text in \cite{ljung_convergence_1978},
where a time delay is assumed. However, as remarked in
\cite{ljung_convergence_1978}, this delay is not essential
for their results. To make clear that algebraic loops are
not allowed in the system, we included an explicit time
delay in the subsystem \eqref{eq:feedback_block}.}.
The assumption on the process $e_k$ implies that 
\eqref{eq:model} is the optimal mean squared error predictor of
$y_k$, given $y_{0:k-1}$ and $u_{0:k}$.

A simple criterion that can be used to obtain estimates
for the parameters in the vector $\theta$ is the minimization 
of the cost function
\begin{equation}
	\label{eq:cost_function} 
	V_N(\theta) = \frac{1}{N}\sum_{k=1}^N 
				\|y_k - \hat{y}_k(\theta)\|^2,
\end{equation}
resulting in the parameter estimates
\begin{equation}
	\label{eq:theta_est} 
	\hat{\theta}_N = \text{arg}\min_{\theta\in \pardomain} V_N(\theta) 
\end{equation}

The asymptotic behavior of the parameter
estimates $\hat{\theta}_N$ as the number of data points $N$
grows to infinity depends on the asymptotic behavior of the
function $V_N(\theta)$. Since the system is stochastic,
$V_N(\theta)$ is a random variable. To guarantee that the
identified model is independent of the specific realization
of the noise entering the system, we need the prediction
error $\varepsilon_k(\theta)=y_k - 
\hat{y}_k(\theta)$ to satisfy an ergodicity property:
$V_N(\theta)$ must converge to its expected value as 
$N\to\infty$. This property is achieved by means of two
fundamental conditions: one on the system that generates
the data, and one on the predictor.

\begin{cond}[\cite{ljung_convergence_1978},\cite{forssell_closed-loop_1999}]
	\label{cond:S3}
	The closed-loop system 
	\eqref{eq:innovations_form}-\eqref{eq:feedback_block}
	is such that for each $k,s \in \setint_+$, $k \ge s$,
	there exist random variables $\bar{y}_{k,s}$ and 
	$\bar{u}_{k,s}$, independent of $\vref_{[0,s]}$ and
	$e_{[0,s]}$ but not independent of $\vref_{[0,k]}$ 
	and $e_{[0,k]}$, such that
	\begin{subequations}
		\label{eq:condition_S3} 
		\begin{align}
			\label{eq:condition_S3_A}
			E\left[\|y_k - \bar{y}_{k,s} \|^4\right] 
			&< C \alpha^{k-s}
			 \\
			\label{eq:condition_S3_B}
			E\left[\|u_k - \bar{u}_{k,s} \|^4\right] 
			&< C \alpha^{k-s}
		\end{align}
	\end{subequations}
	for some $C>0$ and $\alpha < 1$. Here,
	$\bar{y}_{s,s} = \bar{u}_{s,s} = 0$.
\end{cond}

\begin{cond}[\cite{ljung_convergence_1978}]
	\label{cond:M1} 
	The mappings 
	$\hat{F}_k$
	are differentiable with respect to $\theta$ for
	all $\theta \in \pardomain$, where $\pardomain$
	is a closed and bounded subset of $\setreal^\npars$. 
	Furthermore, there exist a $C<\infty$ and $\alpha \in (0,1)$
	such that
	\begin{equation}
		\label{eq:condition_M1_A} 
		\begin{split}
		\big\|\hat{F}_k\left(u_{[0,k]};\theta \right)-&
		\hat{F}_k\left(\tilde{u}_{[0,k]};\theta\right)\big\| \\
			&\le  C\sum_{m=0}^{k} \alpha^{k-m} \;
			\|u_m - \tilde{u}_m\|
		\end{split}
	\end{equation}
	and
	\begin{equation}
		\label{eq:condition_M1_B}
		\big\| \hat{F}_k\left(0_{[0,k]};
			\theta \right) \big\| \le C 
	\end{equation}
	for all $k$, $u_{[0,k]}$, $\tilde{u}_{[0,k]}$, and 
	$\theta$ belongs to an open neighborhood of 
	$\pardomain$. The $(d/d\theta)\hat{F}_k$ are
	subject to an inequality analogous to
	\eqref{eq:condition_M1_A}.
\end{cond}
When the model \eqref{eq:model} satisfies Condition
\ref{cond:M1}, then the mapping $\theta \mapsto 
\{\hat{F}_k(\;\cdot\; ;\theta)\}_{k\in\setint_+}$ is
called a model structure 
\cite[Section 5.7]{ljung_system_1999}. 
Thus \eqref{eq:model} is called a model structure when
viewed as a function of $\theta$.

The main result of \cite{ljung_convergence_1978} can now be
stated as follows.

\begin{lem}[\cite{ljung_convergence_1978}]
	\label{lem:convergence} 
	Consider the feedback system 
	\eqref{eq:innovations_form}-\eqref{eq:feedback_block}
	subject to Condition \ref{cond:S3}, and the model 
	\eqref{eq:model} subject to Condition \ref{cond:M1}.
	Consider $V_N(\theta)$ given by \eqref{eq:cost_function}.
	Then
	\[
		\sup_{\theta\in \pardomain} 
			\left|
				V_N(\theta) - E\left[V_N(\theta)\right]
			\right| \to 0 \quad\quad \text{ w.p. } 1 
			\text{ as } N\to \infty
	\]
\end{lem}

\subsection{Contracting discrete-time dynamics}

Neuronal systems are most commonly represented by 
state-space models, and so we will rely on the state-space
formalism of contraction theory
\cite{lohmiller_contraction_1998} to analyze the
identification problem. We present both the
discrete-time and continuous-time definitions in sequence,
as they are both relevant to us.

First, consider the discrete-time system 
\begin{subequations}
	\label{eq:DT_system}
	\begin{align}
	\label{eq:DT_dynamics}
	 x_{k+1} &= f(x_k,u_k) \\
	 \label{eq:DT_output}
	 y_k &= h(x_k,u_k)
	\end{align}
\end{subequations}
where $f$ and $h$ are continuously differentiable 
functions, $u:\setint_+\to \setreal^\ninputs$
is the input signal, $y:\setint_+\to\setreal^\noutputs$ is
the output signal, and
$x:\setint_+ \to 
\setreal^\nstates$ is the 
state vector. We denote by $x_k=\phi_{k,s}(u,x_s)$ 
the solution of \eqref{eq:DT_dynamics} that starts at time 
$s$ and is evaluated at time $k\ge s$, when 
\eqref{eq:DT_dynamics} is subject to the input sequence
$u=u_{[0,\infty]}$ and initial condition $x_s$. We say a set
$X~\subseteq~\setreal^{\nstates}$ is positively invariant,
uniformly on $U \subseteq \setreal^{\ninputs}$, if 
$\phi_{k,0}(u,x_0) \in X$ for $x_0 \in X$, $u_k \in U$,
and $k\in\setint_+$.

\begin{defn}[\cite{lohmiller_contraction_1998}]
The discrete-time dynamics \eqref{eq:DT_dynamics} is said 
to be exponentially contracting in a set
$X~\subseteq~\setreal^{\nstates}$, uniformly (in $u$) 
on $U \subseteq \setreal^{\ninputs}$, if there exist a
symmetric matrix sequence $P_k(x) \ge \epsilon I > 0 $ 
and a constant $\alpha \in (0,1)$ such that 
\begin{equation}
	\label{eq:DT_contraction} 
	\frac{\partial f }{\partial x}^\top 
	P_{k+1}(f(x,u))
	\frac{\partial f  }{\partial x} 
	\le
	\alpha^2 P_k(x)
\end{equation}
for all $k \in \setint_+$, $x \in X$, and $u \in U$.
\end{defn}

We call $P_k(x)$ the contraction metric, and $\alpha$ the
contraction rate.
The result below will be instrumental in connecting
the contraction property to the PEM conditions of
the previous section. For simplicity, we work with a
constant contraction metric.

\begin{lem}
\label{lem:contraction_expo_stab} 
	Consider the discrete-time system \eqref{eq:DT_system}.
	Let
	\begin{equation}
		\label{eq:output_operator} 
		y_k = F_k(u_{[0,k]};x_0) =
		h(\phi_{k,0}(u,x_0),u_k)
	\end{equation}
	For some $\beta > 0$, assume \eqref{eq:DT_dynamics} is
	exponentially contracting in 
	a positively invariant, convex, closed and bounded set $X$,
	uniformly on $U~=~[-\beta,\beta]^{\ninputs}$, with a constant
	 $P>0$. Then there are $C_1,C_2>0$ and 
	$\alpha \in (0,1)$ such that
		\begin{equation}
			\label{eq:expo_stable_output} 
			\begin{split}
			&\big\|F_k\left(u_{[0,k]};x_0\right)-
			F_k\left(\tilde{u}_{[0,k]};
			\tilde{x}_0\right)\big\| \\
				&\le  C_1 \sum_{m=0}^{k} \alpha^{k-m} 
				\|u_m - \tilde{u}_m\|
				+  C_2 \, \alpha^k \|x_0 - \tilde{x}_0\|
			\end{split}
		\end{equation}
		for all $k\ge 0$, $u,\tilde{u} 
		\in \inpclass^\ninputs$,
		and $x_0,\tilde{x}_0\in X$.
\end{lem}
\begin{proof}
	See the Appendix \ref{proof:contraction_expo_stab}. 
\end{proof}

\subsection{Contracting continuous-time dynamics}

Consider the continuous-time nonlinear system
\begin{equation}
	\label{eq:CT_system} 
	\dot{x}(t) = f(x(t),u(t))
\end{equation}
where $f$ is a continuously differentiable function, 
$u:~\setreal_+\to \setreal^\ninputs$ is 
an input signal and 
$x:\setreal_+ \to 
\setreal^\nstates$ 
is the state vector.

\begin{defn}
\label{def:contractive_system} 
The continuous-time dynamics \eqref{eq:CT_system} is said to be 
exponentially contracting in a set 
$X~\subseteq~\setreal^{\nstates}$, uniformly (in $u$) on 
$U~\subseteq~\setreal^{\ninputs}$, if there exists a continuously
differentiable symmetric matrix $P(x,t) \ge \epsilon I > 0$
and a constant $\lambda > 0$ such that 
\begin{equation}
	\label{eq:CT_contraction} 
	\frac{\partial f}{\partial x}^\top P(x,t) 
	+ P(x,t) \frac{\partial f}{\partial x}  
	+ \dot{P}(x,u,t) \le  -2\lambda P(x,t)
\end{equation}
for all $t \in \setreal_+$, $x \in X$, and $u\in U$.
\end{defn}

Alternatively, by writing $P = \Theta^\top \Theta$,
\eqref{eq:CT_contraction} can be written as
$\tfrac{1}{2}\left( F + F^\top \right) \le -\lambda I$, 
with 
\begin{equation*}
	F = \left(\dot{\Theta} + \Theta \frac{\partial f}
	{\partial x} \right) \Theta^{-1},
\end{equation*}

\section{Conductance-based models under feedback}
\label{sec:models} 

Conductance-based models are biophysical neuronal models
that admit the circuit representation shown in Figure
\ref{fig:conductance_based}. While the framework of the
present paper holds for multiple-input-multiple-output
models, we focus on the single-input single-output case.
Such models were first introduced in the seminal work of
Hodgkin and Huxley \cite{hodgkin_quantitative_1952}. For a
general introduction, the reader is referred to Chapters 3
and 5 in \cite{keener_mathematical_2009}, or
textbooks of neurophysiology such as 
\cite{hille_ionic_1984,izhikevich_dynamical_2007,ermentrout_mathematical_2010}.
To date, conductance-based modeling remains the central
paradigm of biophysical neuronal modeling
\cite{almog_is_2016}. 

Our main results will concern the identification of 
discrete-time stochastic conductance-based models.
However, it is relevant to first introduce these models 
in a continuous-time and deterministic setting 
(Section \ref{sec:CT_models}). This allows us to prove the
\textit{output contraction} property (Section 
\ref{sec:output_feedback}), which is central to our 
results. This property is also satisfied by discrete-time
conductance-based models, which we introduce, along with
the noise setting, at the end of the section.

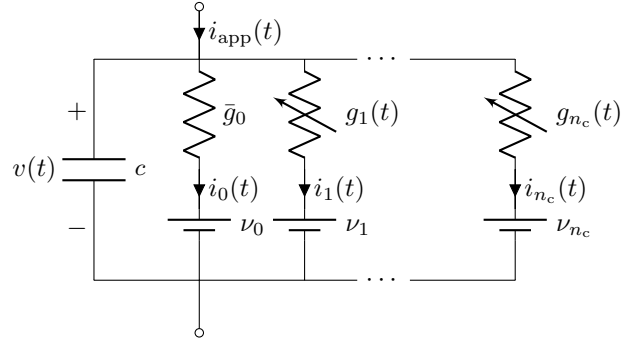
\begin{figure}[t]
	\centering
	\if\usetikz1
	\ctikzset{label/align = smart}
	\begin{tikzpicture}[scale=0.70,
		american resistors,american voltages]
		\newcommand{\cbtop}{4.2}

		\draw (0,0) to[C,l_={\small $c$},v^<=$v(t)\;\;$] 
		(0,\cbtop);
		
		\draw (0,\cbtop) to (5,\cbtop);
		\draw (0,0) to (5,0);
		\node (dots1) at (5.5,\cbtop)[] {$\cdots$}; 
		\node (dots2) at (5.5,0)[] {$\cdots$};
		\draw (6,\cbtop) to (8,\cbtop);
		\draw (6,0) to (8,0);
		
		\draw (2,\cbtop) to[R,l=$\gmax_0$] (2,\cbtop/2) 
			to[battery1,l=$\nu_0$,i>=$i_0(t)$] (2,0);
		\draw (4,\cbtop) to[vR,l=$g_1(t)$] (4,\cbtop/2)
		to[battery1,l=$\nu_1$,i>=$i_1(t)$] (4,0);	
		\draw (8,\cbtop) to[vR,l=$g_{\nchannels}(t)$] 
		(8,\cbtop/2)to[battery1,l=$\nu_\nchannels$,
				i>=$i_\nchannels(t)$] (8,0);
		
		\draw (2,\cbtop+1) 
			to[short,o-,i=$\Iapp(t)$] (2,\cbtop);
		\draw (2,0) to[short,-o] (2,-1);
	\end{tikzpicture}
	\else
		\includegraphics[scale=1]{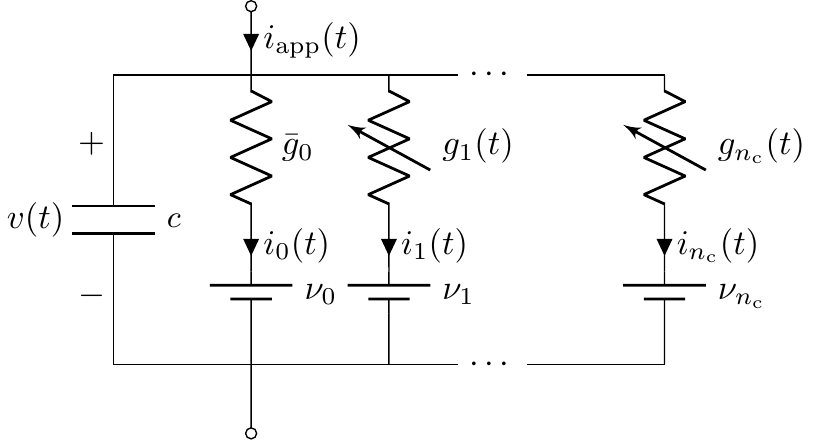}
	\fi
	\caption{Schematic representation of a neuronal
	system.} 
	\label{fig:conductance_based}
\end{figure}

\subsection{Conductance-based models}
\label{sec:CT_models} 

In a conductance-based model, the neuronal membrane 
is modeled by an ideal capacitor of capacitance $c>0$.
The voltage across the membrane, which is the output of 
the model, is given by $v(t) \in \setreal$. The
neuron possesses $\nchannels\in\setnat$ different types of ion
channels embedded in its membrane. These ion channels allow
ionic currents to flow across the membrane according
to Kirchhoff's law,
\begin{equation}
	\label{eq:Kirchhoff} 
	c\,\dot{v}(t) = -\sum_{j=0}^{\nchannels} i_j(t) + \Iapp(t)
\end{equation}
where each current $i_j(t)$, $j=1,\dotsc,\nchannels$, models
an ionic current. The ionic currents not explicitly included in
the model are lumped into a leak  current $i_0(t)$. In addition,
the membrane voltage is affected by an external applied current 
$\Iapp(t)$.

All currents in a conductance-based model obey Ohm's law.
The leak current 
\begin{equation}
	\label{eq:leak_current} 
	i_0(t) = \gmax_0 (v(t)-\nu_0)
\end{equation}
is characterized by a constant conductance $\gmax_0>0$ 
and a constant reversal potential $\nu_0 \in \setreal$. In
contrast,  the conductances of the ionic currents are
voltage-dependent. This dependence is the key source of
nonlinearity of conductance-based models. 
Owing to the original proposal of Hodgkin and Huxley, each 
ionic current has a nonlinear state-space model of the form
\begin{subequations}
	\label{eq:ion_channel}
	\begin{align}
		\label{eq:activation} 
			\tau_{m,j}(v) \, \dot{m}_j & =  
			- m_{j} + m_{\infty,j}(v) 
			\\
		\label{eq:inactivation} 
			\tau_{h,j}(v) \, \dot{h}_{j} & =  
			- h_{j} + h_{\infty,j}(v) 
			\\
		\label{eq:ionic_current} 
		i_j(t) &= \gmax_j  m_{j}(t)^{\alpha_j}
		h_{j}(t)^{\beta_j}(v(t)-\nu_j)
	\end{align}
\end{subequations}
with $j=1,\dotsc,\nchannels$.
The constants $\gmax_j>0$ are called the maximal 
conductances, and $\nu_j\in\setreal$ are called
reversal potentials.
The variables $m_j$ and $h_j$ are called gating variables,
and take values in the closed interval $[0,1]$. Their
dynamics are defined by the continuously differentiable
time-constant functions 
\[\tau_{m,j},\tau_{h,j}:\setreal \to [\tau_{\min},
\tau_{\max}]\subset \setreal_+\] 
and activation functions
\[m_{\infty,j},h_{\infty,j}:\setreal \to 
[0,1]\]
where $\tau_{\min}>0$.
The gating variables modulate the current
conductance with a voltage-dependent first-order lag
dynamics. The exponents $\alpha_j$ and 
$\beta_j$ belong to $\setint_+$, and whenever $\alpha_{j^*} 
= 0$ or $\beta_{j^*}=0$, we ignore \eqref{eq:activation} or
\eqref{eq:inactivation} for $j=j^*$, respectively. These 
exponents, along with the gating variable dynamics 
\eqref{eq:activation}-\eqref{eq:inactivation}, constitute
the \textit{kinetic model} of the $j^{\text{th}}$ ion
channel
\cite{keener_mathematical_2009,hille_ionic_1984}.

A compact representation of the entire model
\eqref{eq:Kirchhoff}-\eqref{eq:ion_channel}  has the 
state-space structure 
\begin{subequations}
	\label{eq:CT_cb_model}
	\begin{align}
		\label{eq:CT_v_dyn}
		c\,\dot{v} &= -\Iint(v,\w) + \Iapp 
		\\
		\label{eq:CT_int_dyn}
		\dot{\w} &= A(v)\w + b(v)
	\end{align}
\end{subequations}
where the vector $\w \in [0,1]^\ngat$ collects all
the gating variables $m_j$ and $h_j$ for which 
$\alpha_j > 0$ and $\beta_j>0$, respectively, and 
\begin{equation}
	\label{eq:i_int} 	
	\Iint(v,\w)=\gmax_0(v-\nu_0) + 
	\sum_{j=1}^{\nchannels} \gmax_j\,m_j^{\alpha_j}\,
	h_j^{\beta_j}(v-\nu_j)
\end{equation}
denotes the total membrane internal current. 
Note that the matrix $A(v)$ is diagonal, and $b(v)$ is
a vector-valued function of $v$. 
The  model \eqref{eq:CT_cb_model}, with input $\Iapp$
and output $v$, is in the standard global normal 
form of nonlinear systems \cite{byrnes_passivity_1991}.
The dynamics \eqref{eq:CT_int_dyn} is called the 
\textit{internal dynamics} of the system. 
An important fact about \eqref{eq:CT_cb_model} is
that the \textit{first} derivative of the output explicitly 
depends on the input, while the internal dynamics does not --- 
in other words, \eqref{eq:CT_cb_model} has a relative degree of
\text{one}. This fact will be explored later on.

\begin{exmp}
	\label{ex:HH} 	
	The Hodgkin-Huxley (HH) model 
	\cite{hodgkin_quantitative_1952} is the prototypical
	conductance-based model. It is given by
	\eqref{eq:Kirchhoff}, with $c=1$ 
	$\mathrm{\upmu F/cm^2}$, and it has 
	two ionic currents ($\nchannels = 2$): a sodium current 
	$i_{\text{Na}}$, and a potassium current 
	$i_{\text{K}}$. It also includes a leak current 
	$i_{\text{L}}$. The currents are given by  
	\begin{equation}
		\label{eq:HH_currents}
		\begin{split}
		i_0&=i_{\text{L}} = 0.3\,(v+54.4) \\
		i_1&=i_{\text{Na}}	= 120\, m_{\text{Na}}^3 
		h_{\text{Na}} (v-55) \\	
		i_2&=i_{\text{K}} = 36 \, m_{\text{K}}^4(v+77) \\
		\end{split}  
	\end{equation}
	The three internal variables are the sodium activation
	$m_1 = m_{\text{Na}}$, sodium inactivation
	$h_1 = h_{\text{Na}}$,
	and potassium activation $m_2 = m_{\text{K}}$ 
	(there is no potassium inactivation in the model, i.e.,
	$\beta_2 = 0$). The
	vector $\w$ collecting these variables is given by
	$\w = (m_1,h_1,m_2)^\top = 
	(m_{\text{Na}},h_{\text{Na}},m_{\text{K}})^\top$.
	The different voltage dependent time-constants 
	$\tau(v)$ and activation functions $m_\infty(v)$ and 
	$h_\infty(v)$ are illustrated in Figure
	\ref{fig:HH_kinetic_functions}, and are detailed 
	in Appendix \ref{app:kinetics}.
	
Each gating variable remains in the interval $[0,1]$, and,  
in the absence of external inputs, the voltage remains in
the interval $[\nu_2,\nu_1]=[\nu_{\text{K}},
\nu_{\text{Na}}]=[-77,55]$. This is illustrated in Figure
\ref{fig:HH_stepExample}, where a \textit{spiking} limit
cycle oscillation occurs in response to a small constant 
input.

\begin{figure}[t]
	\centering
	\if\usetikz1
	\begin{tikzpicture}
		\begin{axis}[height=4.5cm,
			width=4.5cm,
			axis y line = left,
			xlabel={$v \mathrm{[mV]}$}, axis x line = bottom,
			ymin=0,ymax=10,
			ytick={0,5,8.6,10},
			yticklabels={0,5,8.6,10},
			legend pos=north east,
			name = HH_tau,
			]
			\addplot[smooth,thick,color=blue] 
				file{./Data/tau_m.txt};
			\addlegendentry{\tiny $\tau_{m,1}$}
			\addplot[smooth,thick,color=red] 
				file{./Data/tau_h.txt};
			\addlegendentry{\tiny $\tau_{h,1}$}
			\addplot[smooth,thick,color=green] 
				file{./Data/tau_n.txt};
			\addlegendentry{\tiny $\tau_{m,2}$}
		\end{axis}
		\begin{axis}[height=4.5cm,
			width=4.5cm,
			axis y line = left,
			xlabel={$v \mathrm{[mV]}$}, axis x line = bottom,
			ymin=0,ymax=1,
			legend style={at={(0.5,0.45)},anchor=west},
			ytick={0,1},
			anchor = north west,
			at = {(HH_tau.north east)},
			xshift = 1.25cm
			]
			\addplot[smooth,thick,color=blue] 
				file{./Data/sig_m.txt};
				\addlegendentry{\tiny 
					$m_{\infty,1}$} 
			\addplot[smooth,thick,color=red] 
				file{./Data/sig_h.txt};
				\addlegendentry{\tiny 
				$h_{\infty,1}$}
			\addplot[smooth,thick,color=green] 
				file{./Data/sig_n.txt};
			\addlegendentry{\tiny $m_{\infty,2}$}
		\end{axis}
	\end{tikzpicture}
	\else
		\includegraphics[scale=1]{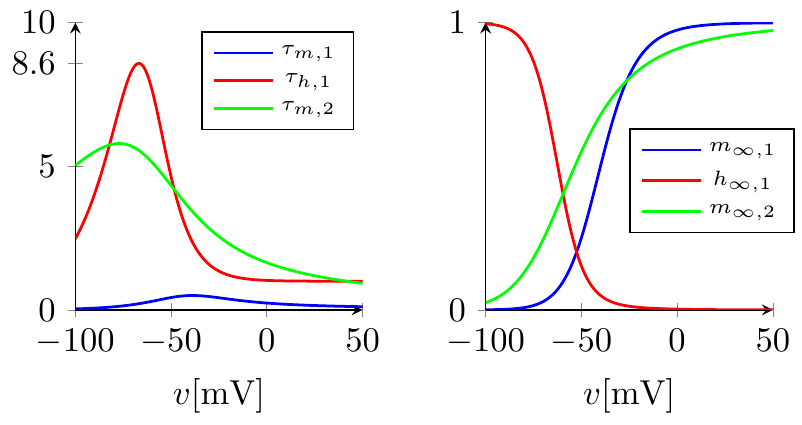}
	\fi
	\caption{
	Left: Time constant functions in the Hodgkin-Huxley model.
	Right: Activation functions in the Hodgkin-Huxley 
	model.}
	\label{fig:HH_kinetic_functions}
\end{figure}
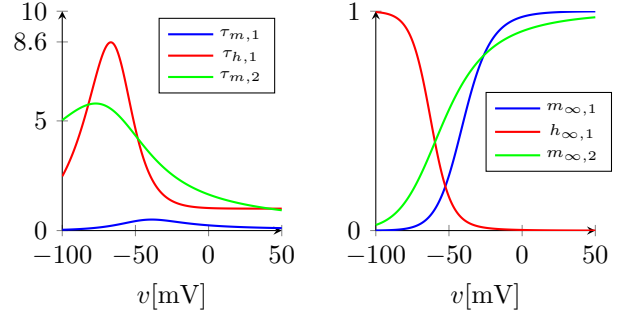

\begin{figure}[b]
	\centering
	\if\usetikz1
	\begin{tikzpicture}
		\begin{axis}[height=4cm,width=8cm,
			ylabel={$v(t)\;\mathrm{[mV]}$}, 
			axis y line = left,
			xlabel={$t\;\mathrm{[ms]}$}, 
			axis x line = bottom, 		
			ymin=-85,ymax=65,
			each nth point={2},
			ytick={-77,0,55},
			yticklabels={$-77$,$0$,$55$},
			name=voltagePlot
			]			
			\addplot [color=blue]
				table[x index=0,y index=1] 
				{./Data/HH_full_stepExample_output.txt};	
		\end{axis}
		\begin{axis}[height=4cm,width=8cm,
			ylabel={$\w(t)$}, 
			axis y line = left,
			xlabel={$t\;\mathrm{[ms]}$}, 
			axis x line = bottom, 		
			ymin=-0.1,ymax=1.1,
			each nth point={2},
			ytick={0,1},
			at={(voltagePlot.below south)},
			anchor=north
			]
			\addplot[color=blue,solid]
				table[x index=0,y index=2] 
			{./Data/HH_full_stepExample_output.txt};
			\addlegendentry{$m_1$}; 
			\addplot[color=green,solid]
				table[x index=0,y index=3] 
			{./Data/HH_full_stepExample_output.txt};
			\addlegendentry{$h_1$};
			\addplot[color=red,solid]
				table[x index=0,y index=4] 
			{./Data/HH_full_stepExample_output.txt};
			\addlegendentry{$m_2$};	
			\end{axis}
	\end{tikzpicture}
	\else
		\includegraphics[scale=1]{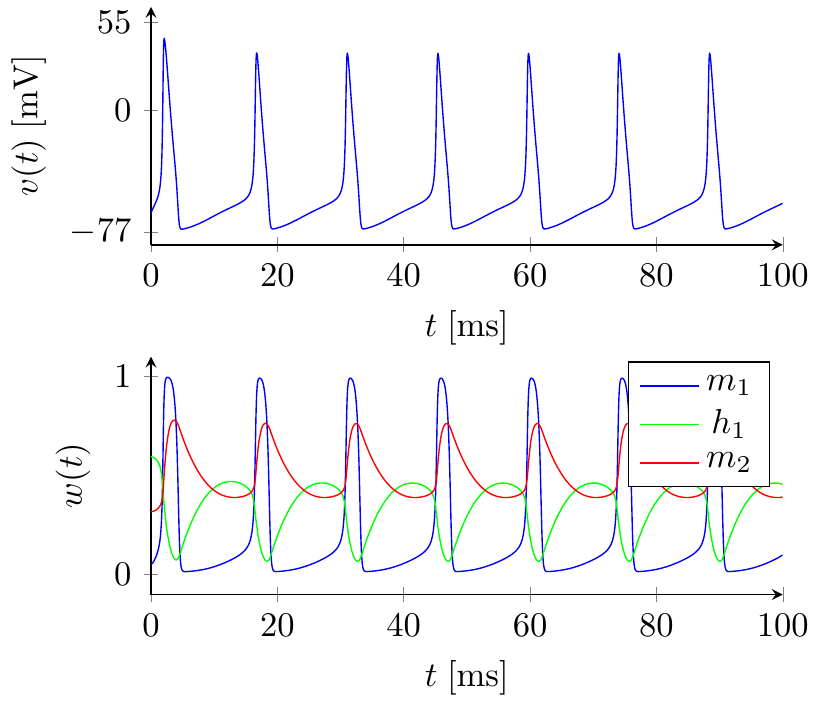}
	\fi
	\caption{Simulated state trajectories of the Hodgkin-Huxley
	model (Example \ref{ex:HH}) for $\Iapp(t) = 10 \;
	\mathrm{\upmu A/cm^2}$.}
	\label{fig:HH_stepExample}
\end{figure}
	 
\end{exmp}

The internal dynamics \eqref{eq:CT_int_dyn} of the 
Hodgkin-Huxley model in
Example \ref{ex:HH} is exponentially contracting in 
$\setreal^3$, uniformly in $v$ on $\setreal$ (see 
Definition \ref{def:contractive_system}). This is verified 
with the constant metric $P=p I$, for any $p>0$, and any 
$\lambda$ such that $0<\lambda<1/\tau_{\max}$: in that case, 
we have
\[
	-2\,p\,\mathrm{diag}\left(\tfrac{1}{\tau_{m,1}(v)},
	\tfrac{1}{\tau_{h,1}(v)},\tfrac{1}{\tau_{m,2}(v)}
	\right)  \le -\tfrac{2}{\tau_{\max}} pI
\]
and we could pick, for instance, $\lambda < 1/8.6$ (see
Figure \ref{fig:HH_kinetic_functions}, left). This is 
in fact a general property of conductance-based models:

\begin{prop}
\label{prop:contractive}
The internal dynamics \eqref{eq:CT_int_dyn} of
\eqref{eq:ion_channel}-\eqref{eq:i_int} is exponentially
contracting in $\setreal^{\ngat}$, uniformly in
$v$ on $\setreal$, i.e., there is a $P_\w>0$ and a $\lambda_\w>0$
such that 
\begin{equation}
	\label{eq:contracting_internal} 
	P_\w A(v) + A(v)^\top P_\w \le -2\lambda_\w P_\w
\end{equation}
for all $v\in\setreal$.
\end{prop}

\subsection{Output feedback contraction}
\label{sec:output_feedback} 

A direct consequence of Proposition~\ref{prop:contractive}
is that a conductance-based model
has a stable inverse. More precisely, using a static output
feedback law, the closed-loop dynamics can be made
exponentially contracting:

\begin{prop}
\label{prop:output_contractive}
Consider a conductance-based model 
\eqref{eq:ion_channel}-\eqref{eq:i_int} subject to the output
feedback law 
\begin{equation}
	\label{eq:CT_output_fb}
	\Iapp(t) = \gain (\vref(t)-v(t)),
\end{equation}
where $\gain > 0$ is a constant gain, and $\vref(t)\in\setreal$ 
is a reference input. Let $\{V_{\gamma}\}$ be a family of closed and
bounded intervals of the real line, 
uniformly in $\gamma>0$. Then, there is a gain $\gamma > 0$
such that the closed-loop dynamics given by 
\eqref{eq:CT_cb_model} and \eqref{eq:CT_output_fb} is
exponentially contracting in $V_\gamma\times [0,1]^\ngat$,
uniformly in $r$ on $\setreal$.
\end{prop}
\begin{proof}

We follow an argument similar to 
\cite[Section 2.2]{wang_partial_2004}. 
The Jacobian (with respect to the states) of the 
closed-loop dynamics \eqref{eq:CT_cb_model},
\eqref{eq:CT_output_fb} is given by
\[
	J = 
		\left[ 
		\begin{array}{cc}
			-\tfrac{1}{c}(\tfrac{\partial \Iint }
			{\partial v}
			+\gamma) &
			-\tfrac{1}{c}
			\tfrac{\partial \Iint}{\partial \w} \\
			\tfrac{\partial A}{\partial v}\w
			+\tfrac{\partial b}{\partial v} &
			A(v)
		\end{array}
	\right]
\]
(we omit dependencies on $\w$ and $v$ for clarity).
By Proposition \ref{prop:contractive}, the internal
dynamics \eqref{eq:CT_int_dyn} has a contraction metric 
$P_\w = \Theta_\w^\top\Theta_\w > 0$ associated with the rate
$\lambda_\w>0$. We will use the matrix
\[
	\Theta = 
	\left[
			\begin{array}{cc}
				c & 0 \\
				0 & \Theta_\w   				
			\end{array}
	\right]
\]
to define a contraction metric $P = \Theta^\top\Theta$ for
the closed-loop system. Define $F = \Theta J \Theta^{-1}$
(this is the generalized Jacobian of the closed-loop
system). Then
\begin{equation*}
	\label{eq:jacobian_interconnection} 
	F = 
	\left[
		\begin{array}{cc}
		F_{11}
		&-\tfrac{\partial \Iint}{\partial \w} \; 
		\Theta_\w^{-1} 	\\ 
		\tfrac{1}{c} \Theta_\w \; 
		\left(\tfrac{\partial A}{\partial v}\w
			+\tfrac{\partial b}{\partial v}\right)
		& F_{22}
		\end{array} 
	\right]
\end{equation*}
with
\begin{equation}
	\label{eq:jacobian_v} 
	F_{11} = - 
	\tfrac{1}{c}\left( \tfrac{\partial \Iint}{\partial v} + \gain \right)
\end{equation}
and
\begin{equation}
	\label{eq:jacobian_q}
	F_{22} = \Theta_\w A(v) 
	\Theta_\w^{-1}
\end{equation}

We will use $F \prec 0$ to denote 
$\tfrac{1}{2}(F+F^\top)\le -\epsilon I$ for all 
$(v,\w^\top)\in V_\gamma\times [0,1]^\ngat$ and some 
$\epsilon > 0$. By Definition~\ref{def:contractive_system},
to demonstrate contraction of the closed-loop system, we
have to show that $F \prec 0$. To do that, we will
require that $F_{11} \prec 0$ and $F_{22} \prec 0$.
Contraction of the internal dynamics 
(Proposition \ref{prop:contractive}) automatically implies
\begin{equation}
\label{eq:F22} 
\tfrac{1}{2} \left( F_{22} + F_{22}^\top 
\right) \le -\lambda_\w I
\end{equation} 
for all $(v,\w^\top)\in\setreal^{\ngat+1}$, 
and thus $F_{22}\prec 0$. Furthermore, 
\begin{equation}
	\label{eq:dhdv} 
\frac{\partial \Iint}{\partial v}(v,\w) = \gmax_0 +
\sum_{j=1}^{\nchannels} \gmax_j m_j^{\alpha_j} h_j^{\beta_j} 
\ge \gmax_0 > 0
\end{equation}
for $\w\in[0,1]^\ngat$ and thus $F_{11} \prec 0$ as well.
Since $F_{11}\prec 0$, by a standard Schur complement
result \cite[pp. 472]{horn_matrix_1985}, we have 
$F \prec 0$ if and only if
\begin{equation}
	\label{eq:cond_gen_jac} 
	\tfrac{1}{2}(F_{22}+F_{22}^\top) 
	<  Q^\top F_{11}^{-1} Q
\end{equation}
where $Q$ is the row vector given by
\begin{equation}
	\label{eq:Q_proof} 
	Q = \frac{1}{2} 
\left(-\tfrac{\partial \Iint}{\partial \w} \; \Theta_\w^{-1} + 
\tfrac{1}{c}
\left(\tfrac{\partial A}{\partial v}\w
			+\tfrac{\partial b}{\partial v}\right)^\top \Theta_\w^\top \right) 
\end{equation}

From \eqref{eq:jacobian_v}, \eqref{eq:dhdv},
\eqref{eq:F22} and
\eqref{eq:cond_gen_jac}, $F\prec 0$ if and only if
\begin{equation}
	\label{eq:contracting_gain_2} 
	\gain I >  
c\left[-\tfrac{1}{2}(F_{22}+F_{22}^\top)\right]^{-1}
	Q^\top Q - \tfrac{\partial \Iint}{\partial v}
\end{equation}

By \eqref{eq:F22} and \eqref{eq:dhdv}, a sufficient
condition for \eqref{eq:contracting_gain_2} to hold is
\begin{equation}
	\label{eq:suf_gen_jac} 
	\gain > \frac{c}{\lambda_\w}\sigma_{\max}[Q]^2
\end{equation}
where $\sigma_{\max}[Q]$ is the largest singular value of $Q$.

Since the continuous functions $\partial \Iint/\partial \w$,
$\partial A/\partial v$ and $\partial b/\partial v$ in
\eqref{eq:Q_proof} are bounded 
on any closed and bounded $V_\gamma\times~[0,1]^\ngat$, it 
follows that 
$\sigma_{\max}[Q]$ is also bounded
on such a set. Since, by assumption,
$V_\gamma$ is uniformly bounded in $\gamma$, a sufficiently large
$\gain$ ensures that \eqref{eq:suf_gen_jac} is satisfied on some
$V_\gamma\times [0,1]^\ngat$.

\end{proof}

The expressions \eqref{eq:contracting_gain_2} and
\eqref{eq:suf_gen_jac} can be used to estimate a lower bound 
on the gain that is necessary to make a conductance-based 
model contracting in a given region of state-space. Depending 
on the choice of the contraction metric $P_\w$, this bound 
can of course be conservative, as illustrated by the 
following example.
\begin{exmp}
	\label{ex:gain} 
	For the Hodgkin-Huxley model 
	(Example~\ref{ex:HH}), choose $P_\w = I$ and $\lambda_\w 
	= 1/8.6 < 1/\tau_{\max} $. Consider the set
	$[-77,55]\times[0,1]^3$.
	Computing
	the right-hand side of \eqref{eq:contracting_gain_2}
	for $v = -77$, $m_1 = h_1 = m_2 = 1$, and 
	$\Theta_\w = I$, leads to the lower bound of
	$ 5.1 \times 10^8$ $\mathrm{mS/cm^2}$ on the 
	gain necessary to ensure exponential contraction of the
	closed-loop	system. Alternatively, consider the
	contraction metric
	$P_\w = 10^6 \times \text{diag}(0.21,3.80,3.16)$.
	In this case, a random search over the set
	$[-77,55]\times[0,1]^3$ gives a less conservative
	lower bound of $2.7 \times 10^3$ $\mathrm{mS/cm^2}$.
\end{exmp}

\subsubsection{Non-ohmic ion currents}

Instead of the Ohmic ionic current
\eqref{eq:ionic_current}, we could have used the
more general formulation
\[
	i_j = \gmax_j  m^{\alpha_j}_{j}
		h^{\beta_j}_{j} p_j(v)
\]
with 
\[
	p_j(v) = \sum_{\ell=0}^{d_j} 
		\eta_\ell \, v^\ell
\]
where each $d_j$ is an arbitrarily large polynomial degree,
and $\eta_\ell\in\setreal$. 
In most non-Ohmic ionic current models, $p_j(v)$ is
a monotonically increasing function  
\cite[Chapter 3]{keener_mathematical_2009}, and the
reversal potential $\nu_j\in\setreal$ is the value where
$p(\nu_j)=0$. In this case, just as in the Ohmic case,
we have
\begin{equation}
	\label{eq:nernst}
	 \text{sign}(p_j(v)) = \text{sign} (v-\nu_j),
\end{equation}
where by convention $\text{sign}\;0 = 0$. 
Since the vast majority of ionic current models is Ohmic, 
we keep the formulation \eqref{eq:ionic_current}, noting 
that all our results can be easily adapted to encompass 
non-Ohmic currents such that \eqref{eq:nernst} holds. 

\subsubsection{The voltage-clamp experiment}
\label{sec:voltage_clamp}

The output contraction property of conductance-based models 
is a consequence of the very experimental protocol that 
has been used to identify neuronal systems in the past:
the voltage-clamp experiment, pioneered by Hodgkin and
Huxley. The voltage-clamp experiment is nothing but a 
high-gain output feedback experimental protocol employed to
stabilize the neuron and to determine its inverse dynamics
through step response experiments. The principle of that
experiment is illustrated in Figure 
\ref{fig:voltage_clamp}. In the limit of high-gain feedback, 
the current drawn from the amplifier to clamp the voltage to the reference
$r(t)$ is by definition the output of the internal dynamics driven by
the voltage $v(t)=r(t)$. Electrophysiologists rely on the stability of that
inverse system to model the internal dynamics through a series
of step responses. In that sense, the contraction property of conductance-based
models is an experimental property of neurons rather than the property of a 
specific mathematical model of the ionic currents.

\begin{figure}[t]
	\centering
	\if\usetikz1
	\resizebox{.8\linewidth}{!}{
	\ctikzset{bipoles/resistor/width=0.3}
	\begin{tikzpicture}[american voltages][scale=1]
		
		\node (CB) 
		[draw,align=center,minimum width=2cm,minimum
		height=3.5cm,anchor=south west,text width=2cm] at
		(0,0)	{Neuronal membrane 
		(Fig. \ref{fig:conductance_based})};
		
		\draw ($(CB.north west)-(4,.5)$) 
		node[op amp,yscale=-1] (amp1) {}
		(amp1.+) node[left ] {$\vref(t)$}
	  	(amp1.out) to[R=$\gmax_e$, i>_=$\Iapp(t)$] 
	  	($(CB.north west)-(0,.5)$);	
	  	
	  	\node at ($(CB.north west)-(4,.5)$) []
	  	{$\bar{\gamma}$};
	  	
		\draw  ($(CB.south west)+(-2,1)$) 
		node[op amp,yscale=-1,xscale=-1] (amp2) {}
		($(CB.south west)+(0,0.5)$) to (amp2.-)
	  	(amp2.out) -| (amp1.-) node[left]{$v(t)$}
	  	($(CB.south west)+(-.5,0.5)$) to
	  	($(CB.south west)+(-.5,0)$)	node[ground]{}
	  	($(CB.north west)+(-.5,-0.5)$) |- (amp2.+);
	  	
	  	\node at ($(CB.south west)+(-2,1)$) []
	  	{$1$};
	\end{tikzpicture}
	}
	\else
		\includegraphics[scale=1]{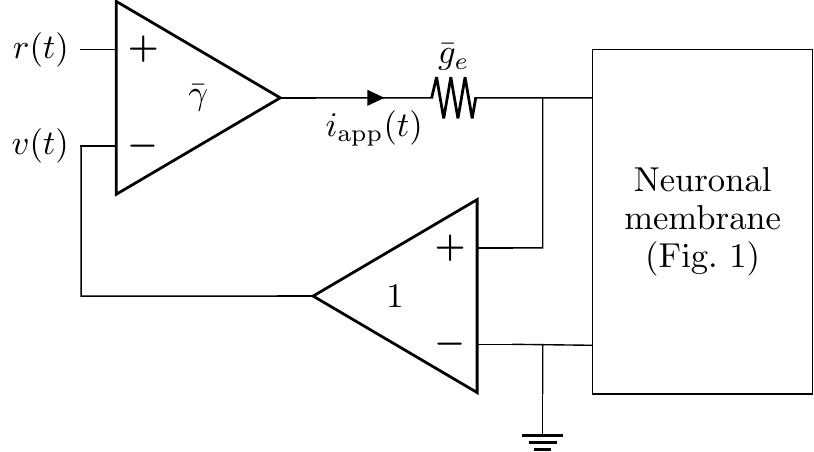}
	\fi
	\caption{The voltage-clamp experiment: electrodes
	are used to inject the current $\Iapp(t)$ and measure
	the voltage $v(t)$ of the neuronal membrane. The 
	amplifiers are ideal differential amplifiers, and
	$\gmax_e$ models the electrode conductance.
	When $\bar{\gamma}\gg 1$, this implements the
	feedback law \eqref{eq:CT_output_fb} with 
	$\gamma =  \bar{\gamma}\gmax_e$.}
	\label{fig:voltage_clamp}
\end{figure}
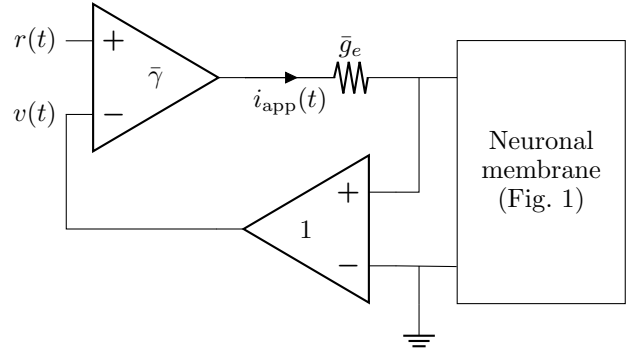

Models of specific ion channel types have been accumulated
over time by electrophysiologists.  Today, online databases
such as ModelDB \cite{mcdougal_twenty_2017} contain large
libraries of ion channels models. The structure of those
models is often used in parametric identification of new
types of neurons (see, e.g.,
\cite{druckmann_novel_2007,huys_efficient_2006}). The
identified parameters include the maximal conductances 
$\gmax_j$ and the Nersnt potentials $\nu_j$.
The purpose of the next sections will be to show that
the classical PEM provides consistent estimates for these
parameters. 

\subsection{Discrete-time stochastic conductance-based models}
\label{sec:DT_CB_models} 

We now turn to the task of identifying a conductance-based 
model from sampled current-voltage data, while taking into
account the intrinsic noise that affects neuronal systems. 
Given a sampling period $\samplingT>0$, we will consider 
the discrete-time stochastic model
\begin{subequations}
\label{eq:DT_CL_model} 
	\begin{align}
		\label{eq:DT_v_system}
			c \tfrac{v_{k+1}-v_k}{\samplingT} &=  
			-\Iint(v_k,\w_k) + \gain(\vref_k-v_k) 
			+ e_k 
				\\
		\label{eq:DT_int_dynamics}
		 \tfrac{\w_{k+1}-\w_k}{\samplingT} &= 
		 A(v_k)\w_k + b(v_k),
	\end{align}
\end{subequations} 
which is a forward-Euler discretization of the closed-loop 
system given by \eqref{eq:ion_channel}-\eqref{eq:i_int} and 
\eqref{eq:CT_output_fb}, with an additive noise
$e_k$ on the input current. The noise current $e_k$ is used 
to model the aggregate effect of ion channel fluctuations 
\cite{goldwyn_what_2011,rowat_interspike_2007} and background 
neuronal activity \cite[Chapter 8]{gerstner_neuronal_2014}.
This system is illustrated in Figure \ref{fig:open-loop_id}.

\begin{figure}[h]
	\centering
	\if\usetikz1
	\begin{tikzpicture}
		\node (x_dyn) at (0,0) [fblock]
				{Eq. \eqref{eq:DT_CL_model}};
		\coordinate (input) at 
			($(x_dyn.west)-(0.5,-0.2)$);
		\coordinate (noise) at 
			($(x_dyn.west)-(0.5,0.2)$);
		\coordinate[right=0.5 of x_dyn] (output);
		\draw[-latex] (input) -- ($(x_dyn.west)-(0,-0.2)$)
			node[left,pos=0]{$\vref_k$};
		\draw[-latex] (noise) -- ($(x_dyn.west)-(0,0.2)$)
			node[left,pos=0]{$e_k$};
		\draw[-latex] (x_dyn) -- (output)
			node[above,pos=1]{$v_k$};
	\end{tikzpicture}
	\else
		\includegraphics[scale=1]{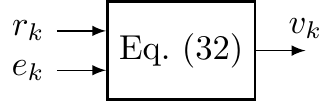}
	\fi
	\caption{Block-diagram of the system \eqref{eq:DT_CL_model}.}
	\label{fig:open-loop_id}
\end{figure}
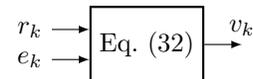

The discretization scheme leading to 
\eqref{eq:DT_CL_model} is classical in system identification 
of biological neurons \cite{huys_efficient_2006}
and in simulations of neuronal behavior
\cite{koch_methods_1989,soudry_conductance-based_2012,baladron_mean-field_2012}.
While more advanced discretization schemes could be 
considered, we stress that both the stochastic discrete-time 
model \eqref{eq:DT_CL_model} and the deterministic 
continuous-time model 
\eqref{eq:ion_channel}-\eqref{eq:i_int} are empirical 
mean-field approximations of the molecular dynamics 
governing the opening and closing of ion channels 
\cite{hille_ionic_1984}. For this reason, one should not
regard \eqref{eq:DT_CL_model} as an {\it approximation} 
of the continuous-time model, but merely as its 
discrete-time stochastic counterpart.
It is also worth noting that an advanced
discretization method tailored for nonlinear systems
in global normal form \cite{yuz_sampled-data_2005} 
cannot improve on the forward-Euler scheme when the 
continuous-time system has a relative degree of one 
--- which is the case for \eqref{eq:CT_cb_model}. 

The remainder of the paper will address the parametric 
identification of \eqref{eq:DT_CL_model}.
The next result shows that when the inputs are bounded, we can
always find a positively invariant set in which discrete-time
conductance-based models can be made contracting by output
feedback:
\begin{prop}
	\label{prop:DT_contraction} 
	Consider the system \eqref{eq:DT_CL_model}.
	Assume that $|\vref_k|,|e_k|<\inpbound$ for all $k \ge 0$; 
	there exist a large enough $\gamma>~0$, a small enough 
	$\samplingT>~0$, and an interval 
	$[v_{\min},v_{\max}]\subset\setreal$ such	that 
	$[v_{\min},v_{\max}]~\times~[0,1]^\ngat$ is a positively
	invariant set for \eqref{eq:DT_CL_model},
	and \eqref{eq:DT_CL_model} is exponentially contracting in 
	$[v_{\min},v_{\max}]~\times~[0,1]^\ngat$, uniformly 
	in $(r,e)$ on $[-\beta,\beta]^2$. Furthermore, there is a 
	small enough $\samplingT>0$ such that $[0,1]^\ngat$ is a positively 
	invariant set for the subsystem \eqref{eq:DT_int_dynamics}, 
	and \eqref{eq:DT_int_dynamics} is exponentially contracting 
	in $[0,1]^\ngat$, uniformly in $v$ on $\setreal$.
\end{prop}
\begin{proof}
	See Appendix \ref{proof:DT_contraction}.
\end{proof}

\section{Identification of neuronal models with the PEM}
\label{sec:identification}

In this section, we discuss the problem of parametric 
identification of the discrete-time stochastic 
conductance-based model \eqref{eq:DT_CL_model}. In Section 
\ref{sec:id_setup}, we frame the problem as one of closed-loop
system identification, and in Section \ref{sec:consistency}, 
we treat the case in which we can consistently identify the 
system's capacitance, maximal conductances and reversal 
potentials.

\subsection{Data-generating system}
\label{sec:id_setup}

Since the data is generated by \eqref{eq:DT_CL_model}, 
we could attempt to identify a discrete-time 
conductance-based model by considering the setup 
shown in Figure \ref{fig:open-loop_id}. However, in that
setup, the \textit{input}-additive noise and the system 
nonlinearities make it difficult to obtain an optimal 
one-step-ahead predictor for the output $v_k$. 
If the noise in the measurements of $v_k$ is negligible, 
we can avoid this issue by viewing \eqref{eq:DT_CL_model}
as a feedback interconnection, and identifying the component 
in the interconnection for which $e_k$ becomes 
\textit{output}-additive noise, and $v_k$ becomes an input. 
This is achieved by partitioning \eqref{eq:DT_CL_model} into
\begin{subequations}
	\label{eq:true_feedback}
	\begin{align}
		\label{eq:true_feedback_v} 
		v_{k+1} &= v_k - \samplingT y_k \\
		\label{eq:true_feedback_u} 
		u_k &= 
		\left[
			\begin{array}{c}
				u_{1,k}		\\
				u_{2,k}
			\end{array}
		\right]
		=
		\left[
			\begin{array}{c}
			\gain(\vref_k-v_k) \\ 
			v_k
			\end{array} 
		\right]				
	\end{align}
\end{subequations}
and
\begin{subequations}
	\label{eq:true_system}
	\begin{align}
		\label{eq:true_system_x} 
		\w_{k+1} &= \w_k + 
		\samplingT 
		\left(\,A(u_{2,k}) \w_k + b(u_{2,k}) \,\right)
		\\
		\label{eq:true_system_y}
		y_k &= \frac{1}{c} 
		(\,
		\sum_{j=1}^{\nchannels} \gmax_j\,m_{j,k}^{\alpha_j}\,
			h_{j,k}^{\beta_j}(u_{2,k}-\nu_j)\\
		\nonumber
			& \hspace{4em}+ \gmax_0(u_{2,k}-\nu_0)	- u_{1,k} - e_k 
		\,) 
	\end{align}
\end{subequations}
where $\w$ collects the states $m_j$ and $h_j$ for
which $\alpha_j,\beta_j~>~0$, and $A(\cdot),b(\cdot)$ 
are determined by 
\eqref{eq:activation}-\eqref{eq:inactivation}.

Most of the dynamics (and any unknown parameters) of 
\eqref{eq:DT_CL_model} are concentrated in the subsystem 
\eqref{eq:true_system}, which has an input $u_k$,
an output $y_k$, and is subject to output-additive
noise $-e_k/c$. It is on the identification of 
\eqref{eq:true_system} that we will focus.
This is a closed-loop identification problem 
(see Section \ref{sec:PEM}): 
in particular, if
$\phi^{w}_{k,0}(u_2,w_0)$ is the solution of
\eqref{eq:true_system_x},
the signal $y_k$
can be written in the form
\eqref{eq:innovations_form},
with
\begin{equation}
	\label{eq:true_system_operator} 
	F_k\left(u_{[0,k]};w_0\right) 
		= \tfrac{1}{c}\left(\,g(u_{2,k},\phi^{w}_{k,0}(u_{2},w_0)\,) 
		- u_{1,k}\right) 
\end{equation}
where $g(\cdot,\cdot)$ is given by \eqref{eq:i_int}.
Similarly, $u_k$ can be written in the form
\eqref{eq:feedback_block}.
This leads to the setup in Figure 
\ref{fig:closed-loop_id}.

\begin{assum}
	\label{asp:discrete_input} 
	The noise $e_k$ in \eqref{eq:true_system_y} 
	is a sequence of	independent	random variables 
	with $E[e_k] = 0$ and a finite variance 
	$E[e_k^2] = \sigma_e^2 > 0$.
	All realizations of $e$ belong to $\inpclass$.
\end{assum}

\begin{assum}
	\label{asp:v_noise} 
	The signals $\vref_k$ and $v_k$ are exactly known
	(no measurement noise), and $\vref$ is a deterministic 
	signal that belongs to $\inpclass$.
\end{assum}

Assumption \ref{asp:v_noise} is consistent with the voltage-clamp
experiment: it allows for current noise but assumes that the
voltage is perfectly measured. 

\begin{assum}
	\label{asp:contracting_true_system}
	The closed-loop dynamics
	\eqref{eq:true_feedback}-\eqref{eq:true_system} is
	exponentially contracting in a positively invariant set
	$[v_{\min},v_{\max}]	\times[0,1]^{\ngat}$,
	uniformly on $[-\beta,\beta]^2$. 	Additionally, 
	$(v_0,\w_0^\top) \in [v_{\min},v_{\max}]	
	\times[0,1]^{\ngat}$. 
\end{assum}

It follows directly from Proposition~\ref{prop:DT_contraction}
that Assumption~\ref{asp:contracting_true_system} can
always be verified\footnote{There is a tradeoff in the choice of
the values of $\gain$ and $\samplingT$, which is made clear in 
the proof of Proposition \ref{prop:DT_contraction}. Increasing
the value of $\gain$ might require decreasing the value of 
$\samplingT$ so that contraction of the discrete-time system is
preserved.} for large enough 
$\gamma>0$ and small enough $\samplingT>0$. Notice that since the 
system parameters are unknown prior to identification, in
practice we cannot check contraction of the closed-loop system by 
direct calculations. An experimental alternative is to probe the 
system and check whether input-output properties implied by 
contraction are verified. Such properties include exponential 
convergence to a unique equilibrium point, under constant input 
(implied by Lemma \ref{lem:contraction_expo_stab}), and 
entrainment by periodic inputs \cite[Theorem 2]{russo_global_2010}. 
We will return to this point in Section \ref{sec:examples}. 

\begin{lem}
\label{lem:S3} 
Under 
Assumptions \ref{asp:discrete_input}-\ref{asp:contracting_true_system}, 
the closed-loop system \eqref{eq:true_feedback}-\eqref{eq:true_system}
satisfies Condition \ref{cond:S3}.
\end{lem}
\begin{proof}
See the Appendix \ref{proof:S3}.
\end{proof}

\begin{figure}[t]
	\centering
	\if\usetikz1
	\begin{tikzpicture}
		\node (v_dyn) at (0,0) [fblock]
		{$F_k\left(u_{[0,k]};w_0\right)$};
		\node (int_dyn) at (0,-1.5) [fblock,
			 align=left]
		{Eq. \eqref{eq:true_feedback}};
		\node (sum_cl) [sum,right=0.5 of v_dyn,
				label={185:$+$},label={5:$-$}]{};
		\coordinate[left= 0.5 of v_dyn] (sum1);		
		\coordinate (input) at 
			($(int_dyn.east)+(0.5,-0.2)$);	
		\coordinate[right=0.5 of sum_cl] (noise);
		\draw[-latex] (noise) -- (sum_cl) 
			node[right,pos=0]{$\tfrac{1}{c}\Inoi_k$};
		\draw[-latex] (input) -- 
			($(int_dyn.east)+(0,-0.2)$)
			node[right,pos=0]{$\vref_k$};
		\draw[-latex] (v_dyn) -- (sum_cl);
		\draw[-latex] (sum_cl) |- 
			($(int_dyn.east)+(0,0.2)$)
			node[right,pos=0.2]{$y_k$};
		\draw[-latex] (int_dyn) -| (sum1)
			node[left,pos=1]{$u_k$} -- (v_dyn);
	\end{tikzpicture}
	\else
		\includegraphics[scale=1]{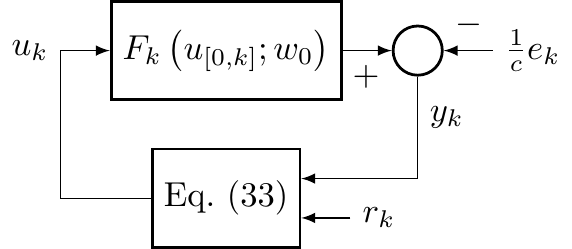}
	\fi
	\caption{Block-diagram of the feedback system
	\eqref{eq:true_feedback}-\eqref{eq:true_system}.
	The mappings $F_k$ are given by 
	\eqref{eq:true_system_operator}.
	}
	\label{fig:closed-loop_id}
\end{figure}
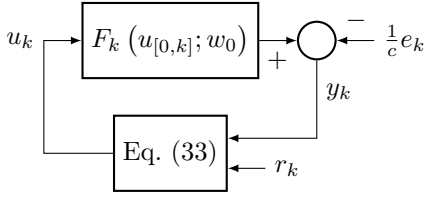

\subsection{Identification with fixed ion 
channel kinetics}
\label{sec:consistency}

Recall that the dynamics \eqref{eq:true_system_x}, 
as well as the exponents $\alpha_j$, $\beta_j$ in 
\eqref{eq:true_system_y}, are determined by ion channel
kinetic models. Given a library of known ion channel 
kinetic models, we will concentrate on identifying 
the parameters $c$, $\gmax_j$, and $\nu_j$ in 
\eqref{eq:true_system_y}, for $j=0,1,\dotsc,\nchannels$. 
This can be achieved by postulating a predictor model 
containing $\nmodel\ge\nchannels$ known ion channel 
kinetic models, chosen a priori. 

For $j=1,\dotsc,\nmodel$, let the predictor states 
be given by $\hat{m}_j$ and $\hat{h}_j$; to each 
of these states, we associate the exponents 
$\hat{\alpha}_j\in\setint_+$ and 
$\hat{\beta_j}\in\setint_+$, respectively. 
We define the predictor by 
\begin{subequations}
	\label{eq:predictor}
	\begin{align}
		\label{eq:predictor_x} 
		\hat{\w}_{k+1} &= \hat{\w}_k + \samplingT 
		(\hat{A}
		(u_{2,k})\hat{\w}_k + \hat{b}(u_{2,k}))
				\\
		\label{eq:predictor_y}
		\hat{y}_k(\theta) &= 
		\sum_{j=1}^{\nmodel}
		\hat{m}_{j,k}^{\hat{\alpha}_j} \,
		\hat{h}_{j,k}^{\hat{\beta}_j} \,
		(\theta_j^{(1)}
		+ 
		\theta_j^{(2)}
		u_{2,k})
		\\
		\nonumber
		& \hspace{4em} + 
		\theta_0^{(1)}
		+
		\theta_0^{(2)}
		u_{2,k} + \theta^{(3)} u_{1,k}
	\end{align}
\end{subequations}
where $u_{1,k}$ and $u_{2,k}$ are given by 
\eqref{eq:true_feedback_u}, the vector 
$\hat{\w}$ collects the gating variables $\hat{m}_j$ and
$\hat{h}_j$ for which $\hat{\alpha}_j>0$ and
$\hat{\beta_j}>0$, respectively, and
$\theta^{(1)},\theta^{(2)}\in\setreal^{\nmodel\times 1}$ 
and $\theta^{(3)}\in\setreal$ are predictor parameters.
The predictor states evolve analogously to the 
(forward-Euler) discretized version of 
\eqref{eq:activation}-\eqref{eq:inactivation}, but 
with activation and time constant functions given by
$\hat{\tau}_{m,j}(\cdot)$, 
$\hat{\tau}_{h,j}(\cdot)$, 
$\hat{m}_{\infty,j}(\cdot)$ and
$\hat{h}_{\infty,j}(\cdot)$.

Comparing \eqref{eq:true_system_y} with 
\eqref{eq:predictor_y}, we see that the predictor 
parameters $\theta_j^{(1)}$, $\theta_j^{(2)}$ and
$\theta^{(3)}$ are meant to identify expressions
involving the true system parameters $c$, $\gmax_j$ 
and $\nu_j$. To formalize this, we make the following
assumption:

\begin{assum}
	\label{asp:model_structure} 
	The model structure \eqref{eq:predictor} 
	contains the true system 
	\eqref{eq:true_system}: we have
	$\hat{\tau}_{m,j}=\tau_{m,j}$, 
	$\hat{\tau}_{h,j}=\tau_{h,j}$, 
	$\hat{m}_{\infty,j}=m_{\infty,j}$, 
	$\hat{h}_{\infty,j}=h_{\infty,j}$, 
	$\hat{\alpha}_j=\alpha_j$ and $\hat{\beta}_j=\beta_j$
	for $j=1,\dotsc,\nchannels \le \nmodel$.
	Additionally, the dynamics 
	\eqref{eq:predictor_x} is exponentially 
	contracting in the positively invariant 
	set $[0,1]^{n_{\hat{w}}}$, uniformly in $v$ on
	$\setreal$. 	
\end{assum}

Again, it follows from Proposition~\ref{prop:DT_contraction}
that Assumption \ref{asp:model_structure} can always
be verified for small enough $\samplingT>0$.
Under Assumption \ref{asp:model_structure}, we now 
see that the true parameter vector, denoted by
$\bar{\theta} = 
	(\bar{\theta}^{(1)\top},\bar{\theta}^{(2)\top},
	\bar{\theta}^{(3)})^\top$,
is given by
\begin{equation}
\label{eq:true_parameters} 
	\begin{split}
	\bar{\theta}_j^{(1)} &=
	\left\{
	\begin{array}{ll}
		-\gmax_j\nu_j/c, &\quad  j=0,1,\dotsc,\nchannels	\\
		0, & \quad  j > \nchannels
	\end{array}
	\right.		
	\\
	\bar{\theta}_j^{(2)} &=
	\left\{
	\begin{array}{ll}
		\gmax_j/c, &\quad j=0,1,\dotsc,\nchannels	\\
		0, &\quad  j > \nchannels
	\end{array}
	\right.  \\
	\bar{\theta}^{(3)} &= -1/c
	\end{split}
\end{equation}

To simplify our results, we will assume the following:

\begin{assum}
	\label{asp:same_initial_conditions} 
	The initial states of the true dynamics
	\eqref{eq:true_system_x} and of
	the predictor dynamics 
	\eqref{eq:predictor_x} satisfy 
	$\hat{m}_{j,0}=m_{j,0}$,  
	$\hat{h}_{j,0}=h_{j,0}$ for 
	$1\le j \le \nchannels$, and
	$\hat{m}_{j,0},\hat{h}_{j,0} \in [0,1]$
	for $\nchannels<j\le \nmodel$.
\end{assum}

As long as Assumptions 
\ref{asp:contracting_true_system} and
\ref{asp:model_structure} are verified, 
due to the contraction property, 
Assumption \ref{asp:same_initial_conditions} can be 
ensured in practice by discarding initial segments 
of the data.

Under Assumptions \ref{asp:discrete_input}, 
\ref{asp:v_noise}, \ref{asp:model_structure} and
\ref{asp:same_initial_conditions}, 
\eqref{eq:predictor} is the optimal mean squared error 
one-step-ahead predictor of $y_k$ in \eqref{eq:true_system}.
The closed-loop identification approach thus avoids 
the intractability in the computation of an optimal
predictor for the forward dynamics output $v_k$. 
	
	Collecting the parameters in a single vector 
	$\theta\in\setreal^{2\nchannels+1}$ given by
\[
	\theta = (\theta^{(1)\top},\theta^{(2)\top},
	\theta^{(3)})^\top,
\]
	we can more compactly write \eqref{eq:predictor_y} as
	\[
		\hat{y}_k(\theta) = \psi_k \theta
	\] 
	with the
	row vector $\psi_k\in\setreal^{1\times(2\nmodel+1)}$
	 given by
	\begin{equation*}
	\begin{aligned}
		\psi_k =& 
		\left(1,\;
		\hat{m}_{1,k}^{\hat{\alpha}_{1}} \,
		\hat{h}_{1,k}^{\hat{\beta}_{1}}\;,
		\; \dotsc\;, \;
		\hat{m}_{\nmodel,k}^{\hat{\alpha}_{\nmodel}} \,
		\hat{h}_{\nmodel,k}^{\hat{\beta}_{\nmodel}}\;,\;
		u_{2,k} \; ,
		\right.
		\\
		&
		\left. u_{2,k} \,	
		\hat{m}_{1,k}^{\hat{\alpha}_{1}} \,
		\hat{h}_{1,k}^{\hat{\beta}_{1}},
		\;\dotsc\;,\;
		u_{2,k} \;
		\hat{m}_{\nmodel,k}^{\hat{\alpha}_{\nmodel}} \,
		\hat{h}_{\nmodel,k}^{\hat{\beta}_{\nmodel}},\;
		u_{1,k} \right)
		\end{aligned}
	\end{equation*}
	Gathering $\psi_k$ in a matrix 
	$\Psi_N\in\setreal^{N\times (2\nmodel+1)}$ given by
	\begin{equation}
		\label{eq:regressor_matrix} 
		\Psi_N = \left[ \psi_N^\top,\psi_{N-1}^\top,\dotsc,
		\psi_1^\top \right]^\top
	\end{equation}
	we find that the vector of model structure outputs
	from time $k=N$ down to time $k=1$ is given by 
	\[
		\hat{y}_{[1,N]}^\top(\theta) =
		\Psi_N \theta
	\]
The above formulation shows that the ion channel kinetic
models act as basis operators mapping the input sequence 
$u_{[0,N]}$, given by \eqref{eq:true_feedback_u}, to the
columns of $\Psi_N$. 
	
\begin{assum}[Persistency of excitation]
	\label{asp:persistency_excitation}
	There is a $N^*>0$ such 
	that $\tfrac{1}{N}\Psi_N^\top \Psi_N$ and 
	$E\left[\tfrac{1}{N}\Psi_N^\top \Psi_N \right]$ are 
	positive-definite for all $N>N^*$.
\end{assum}

Assumption \ref{asp:persistency_excitation} is an 
assumption both on the model structure and on $\vref_k$, the
signal used to excite the true system. Intuitively,
we should not include two identical ion channel kinetics in
the model structure, and the excitation signal $\vref_k$ should
be sufficiently rich.

Under Assumption \ref{asp:v_noise}, we are able to compute
\begin{equation*}
	y_k = -\frac{v_{k+1}-v_k}{\samplingT}
\end{equation*}
from the measurements, and thus we can form the cost
function $V_N(\theta)$ given by \eqref{eq:cost_function}.
There is practical relevance in the fact that a single 
forward difference of the voltage yields $y_k$, which is a 
consequence of the relative degree one property of neuronal
models. 
We can now state the main identification result:

\begin{thm}
	\label{thm:consistency} 
	Let Assumptions	
	\ref{asp:discrete_input}-\ref{asp:same_initial_conditions} be 
	satisfied. Let $N>N^*$, and let 
	$\hat{\theta}_N = (\hat{\theta}_N^{(1)\top},
	\hat{\theta}_N^{(2)\top},\hat{\theta}_N^{(3)})^\top$ 
	be given by
	\begin{equation}
		\label{eq:parameter_estimate} 
		\begin{split}
				\hat{\theta}_N 
			&= \arg\min_{\theta \in \pardomain} 
			V_N(\theta)\\
			&=\arg\min_{\theta \in \pardomain} 
			\frac{1}{N}		
			\|y_{[1,N]}^\top - \Psi_N\theta\|^2
			\end{split}
	\end{equation}		
	where $y_k$ and $\Psi_N$ are given by
	\eqref{eq:true_system_y} and
	\eqref{eq:regressor_matrix}, respectively, and 
	$\pardomain$ is a compact parameter domain containing
	$\bar{\theta}$, the true parameter vector
	\eqref{eq:true_parameters}.
	Then, we have $\hat{\theta}_{N} \to \bar{\theta}$
	w.p. 1 as $N\to\infty$.

\end{thm}

\begin{proof}
	By Assumptions \ref{asp:model_structure}
	and	\ref{asp:same_initial_conditions}, the true 
	output $y_k$, given by \eqref{eq:true_system_y}, can be
	written as
	\[
		y^\top_{[1,N]} = \Psi_N \bar{\theta} -
		 \tfrac{1}{c}e^\top_{[1,N]}
	\]
	and thus we can write
	\[
		E\left[ V_N(\theta) \right] = 
		\tfrac{1}{N} E \left[ 
		\|\Psi_N 
		(\bar{\theta}-\theta) - 
		\tfrac{1}{c}e^\top_{[1,N]} 
		\|^2\right]
	\]
	By Assumption \ref{asp:discrete_input}, the
	time-delay present in the system ensures that $v_k$ and
	$\w_k$ do not depend on $e_k$. We then have that 
	\[
		E \left[ \Psi_N^\top e^\top_{[1,N]} \right] = 0
	\]
	and thus
	\[
		E\left[ V_N(\theta) \right] = \tfrac{1}{N}
		(\bar{\theta}-\theta)^\top E\left[ 
		\Psi_N^\top \Psi_N \right]
		(\bar{\theta}-\theta) + \tfrac{1}{c}\sigma_e^2
	\]	
	Using Assumption \ref{asp:persistency_excitation}, 
	we have
	\begin{equation}
		\label{eq:min_expected_cost} 
		\arg \min_{\theta\in\mathcal{D}} 
		E\left[ V_N(\theta) \right]
		= \bar{\theta}
	\end{equation}
	for all $N>N^*$.
	
	It remains to show that $\hat{\theta}_N$ converges to
	\eqref{eq:min_expected_cost} w.p. 1 as $N\to\infty$.
	This is done by verifying Conditions 
	\ref{cond:S3} and \ref{cond:M1} of 
	Lemma \ref{lem:convergence}.	Condition \ref{cond:S3} is 
	satisfied due to Lemma \ref{lem:S3}.
	By Assumptions \ref{asp:v_noise} and
	\ref{asp:contracting_true_system}, 
	$(v_k,r_k)$ remains in the bounded set
	$[v_{\min},v_{\max}]\times[-\beta,\beta]$, and thus
	the predictor input \eqref{eq:true_feedback_u} 
	belongs to $\mathcal{U}_{\beta^*}^2$ for some 
	$\beta^*>0$. By Assumptions \ref{asp:contracting_true_system} and
	\ref{asp:same_initial_conditions}, 
	$\hat{w}_0 \in [0,1]^{n_{\hat{w}}}$. It follows
	by Assumption \ref{asp:model_structure} and 
	Lemma \ref{lem:contraction_expo_stab} that the predictor 
	\eqref{eq:predictor} verifies Condition 
	\ref{cond:M1}. Finally, Lemma \ref{lem:convergence} 
	ensures 	that $V_N(\theta)$ converges uniformly to 
    $E\left[V_N(\theta)\right]$ on the compact set 
    $\pardomain$. In view of 
    \eqref{eq:parameter_estimate} and 	
    \eqref{eq:min_expected_cost}, this ensures the result
	of the theorem.
\end{proof}

An immediate consequence of Theorem \ref{thm:consistency} is
that we are able to obtain consistent estimates $\hat{c}$, 
$\hat{\gmax}_j$ and $\hat{\nu}_j$ of the original unknown 
parameters of the system \eqref{eq:true_system_y}, for
$j=0,1,\dotsc,\nchannels$. They can be recovered from
$\hat{c}_N =-1/\hat{\theta}^{(3)}_N$,
$\hat{\nu}_{j,N} = -\hat{\theta}^{(1)}_{j,N}/
	\hat{\theta}^{(2)}_{j,N}$
and $\hat{\gmax}_{j,N} = 
-\hat{\theta}^{(2)}_{j,N}/\hat{\theta}^{(3)}_{j,N}$.

\section{Examples}
\label{sec:examples} 

In this section, we illustrate the results of Section 
\ref{sec:consistency} by identifying various discrete-time
neuronal models. All discrete-time models are obtained by
forward-Euler discretization of their continuous-time
counterparts with $\samplingT = 0.005$ ms.

\begin{exmp}
\label{ex:id_HH} 

In this example, we identify the discrete-time Hodgkin-Huxley
(HH) model 
\begin{equation*}
	\label{eq:HH_DT} 
	\begin{split}
	\tfrac{v_{k+1}-v_k}{\samplingT} &= 
	-0.3(v_k+54.4) -\sum_{j=1}^2 i_{j,k} 
	+ \gain(\vref_k-v_k) + e_k \\
	i_{1,k} &= 120 \, m_{1,k}^3 h_{1,k} (v_k-55) \\
	i_{2,k} &= 36 \, m_{2,k}^4 (v_k+77)
	\end{split}
\end{equation*}
where the states $m_j$ and $h_j$ are given by the forward-Euler
discretization of \eqref{eq:activation} and
\eqref{eq:inactivation}, respectively, with activation and 
time-constant functions as in Example~\ref{ex:HH}.
We include in the model structure the two ion channel 
kinetics present in the true model, and identify
the values of $c$, $\gmax_j$, and $\nu_j$ using the
parameter vector $\theta$. Comparing the above expression to
\eqref{eq:predictor_y}, we have the true parameters shown in
Table \ref{tab:HH_pars}.

\begin{table}[b]
\centering
\begin{tabular}{|c|c|c|c|c|c|c|}
\hline
$\bar{\theta}_0^{(1)}$ & $\bar{\theta}_0^{(2)}$ &
$\bar{\theta}_1^{(1)}$ & $\bar{\theta}_1^{(2)}$ &
$\bar{\theta}_2^{(1)}$ & $\bar{\theta}_2^{(2)}$ &
$\bar{\theta}^{(3)}$
\\
\hline
$0.3 \cdot 54.4$ & $0.3$ & $120 \cdot -55$ & $120$ &
$36\cdot 77$ & $36$ & $-1$\\
\hline
\end{tabular}
\caption{True parameters of the Hodgkin-Huxley model.}
\label{tab:HH_pars} 
\end{table}

As mentioned in Section \ref{sec:id_setup}, contraction 
of the closed-loop dynamics can be verified empirically.
Figure \ref{fig:HH_testContraction} (top) 
illustrates the contraction observed for a gain 
of $\gamma~=~50$ in a series of step response 
experiments where the reference $\vref_k$ is first set to 
different baseline values, and then stepped to the same 
final value. In the case shown in Figure 
\ref{fig:HH_testContraction} (top), the voltage $v_k$ converges 
to the same steady-state\footnote{Because of input noise, 
the voltage actually oscillates randomly inside a small 
interval.} no matter what the initial value was at 
$t=10\;\mathrm{ms}$. 

To identify the HH model, we simulated a $5$-second long 
data-gathering experiment with the gain $\gain =~50$. The
reference is $\vref_k = -45 + \tilde{\vref}_k$, where 
$\tilde{\vref}_k$ is white Gaussian noise of standard deviation 
$\sigma_r = 100\,\mathrm{mV}$ that is first filtered by the 
zero-order hold discretization of the system $10^2/(s+10)^2$, 
then truncated so that $|\tilde{r}_k|\le 100$ for $k\ge 0$. 
The input noise $e_k$ is white Gaussian noise of $\sigma_e = 2.5 
\, \mathrm{\upmu A/cm^2}$ that is truncated so that $|e_k| \le 20$ 
for $k\ge 0$. This setup resulted in a signal-to-noise ratio 
(between $y_k$ and $e_k$) of around $30.8$ dB.
A $100\;\mathrm{ms}$ sample of the output $v_k$ used for 
identification is shown in Figure \ref{fig:HH_testContraction} 
(bottom). Notice that the step experiments (top)
explore a voltage interval similar to that explored in 
the data-gathering experiment (bottom).

\begin{figure}[t]
	\centering
	\if\usetikz1
	\begin{tikzpicture}
		\begin{axis}[height=4cm,width=8cm,
			ylabel={$v_k\;\mathrm{[mV]}$}, 
			axis y line = left,
			xlabel={$k\samplingT\;\mathrm{[ms]}$}, 
			axis x line = bottom, 		
			xmin=0,xmax=30,
			ymin=-85,ymax=0,
			each nth point={2},
			name=vclamp_steps
			]			
			\addplot [color=red]
				table[x index=0,y index=2] 
				{./Data/HH_uniquess_20.txt};	
			\addplot [color=green]
				table[x index=0,y index=2] 
				{./Data/HH_uniquess_0.txt};							
			\addplot [color=cyan]
				table[x index=0,y index=2] 
				{./Data/HH_uniquess_-20.txt};
			\addplot [color=magenta]
				table[x index=0,y index=2] 
				{./Data/HH_uniquess_-40.txt};
			\addplot [color=black]
				table[x index=0,y index=2] 
				{./Data/HH_uniquess_-60.txt};
			\addplot [color=blue]
				table[x index=0,y index=2] 
				{./Data/HH_uniquess_-80.txt};
		\end{axis}
		\begin{axis}[height=4cm,width=8cm,
			ylabel={$v_k\;\mathrm{[mV]}$}, 
			axis y line = left,
			xlabel={$k\samplingT\;\mathrm{[ms]}$}, 
			axis x line = bottom, 			
			xmin=100,xmax=200,
			ymin=-85,ymax=0,
			each nth point={5},
			at={(vclamp_steps.below south)},
			anchor=north
			]
			\addplot[color=blue,solid]
				table[x index=0,y index=1] 
					{./Data/HH_iodata.txt};
			\end{axis}
	\end{tikzpicture}
	\else
		\includegraphics[scale=1]{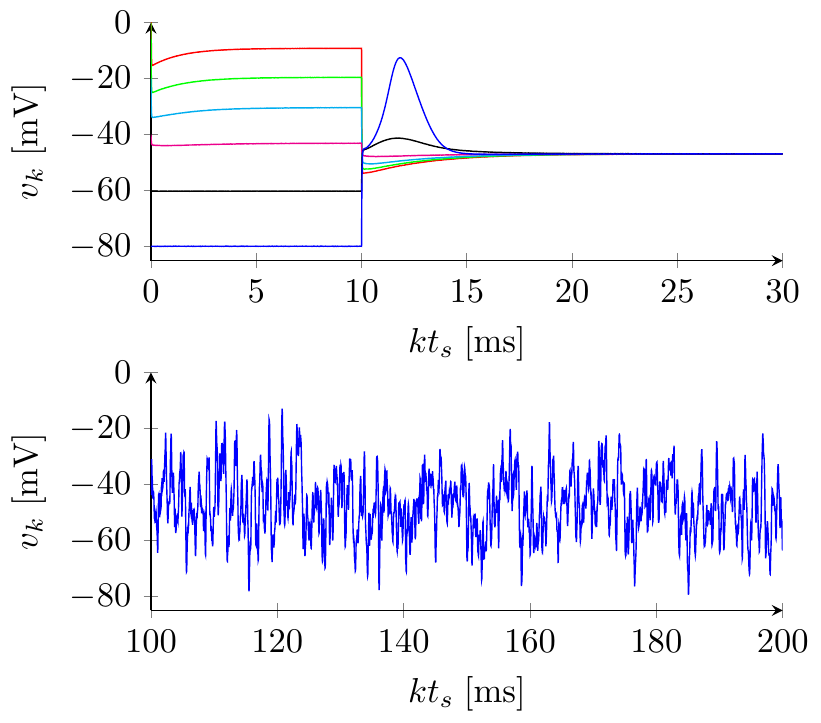}
	\fi
	\caption{Voltage output $v_k$ of the discrete-time
	Hodgkin-Huxley model identified in Example
	\ref{ex:id_HH} subject to different reference inputs $\vref_k$. 
	Top: six experiments in which $\vref_k$ is first set to
	different baseline values ($-80,-60,-40,-20,0$ and $20\;\mathrm{mV}$)
	and then stepped (at $10\;\mathrm{ms}$) to the same
	final value ($-45\;\mathrm{mV}$).
	Bottom: a $100\;\mathrm{ms}$ sample of the voltage output
	used for identification.}
	\label{fig:HH_testContraction}
\end{figure}
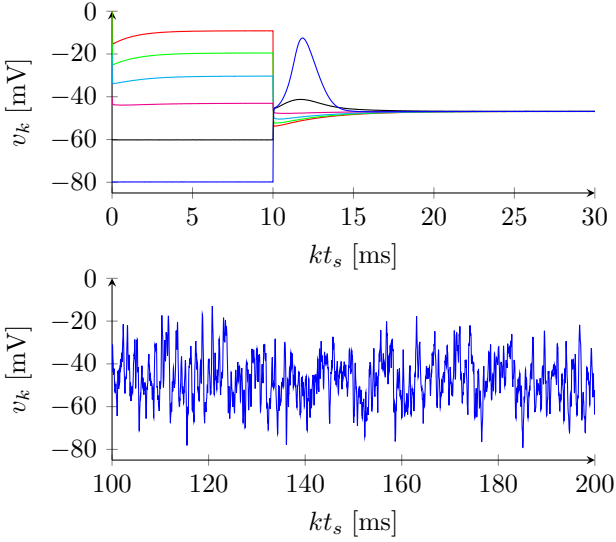

To eliminate transient effects and satisfy Assumption 
\ref{asp:same_initial_conditions} as close as possible, we
eliminated the initial $0.5$ seconds of measurement (corresponding
to $10^5$ samples) from all datasets. Figure 
\ref{fig:HH_estimation} shows the resulting estimation error 
$\bar{\theta}-\hat{\theta}_N$ for $N = 10^5$ to $N=9\times 10^5$
($0.5$ to $4.5$ seconds) for 20 different realizations of 
the experiment, as well as their average; we can see from the 
figure that the estimates steadily converge to the true 
parameters.

\end{exmp}

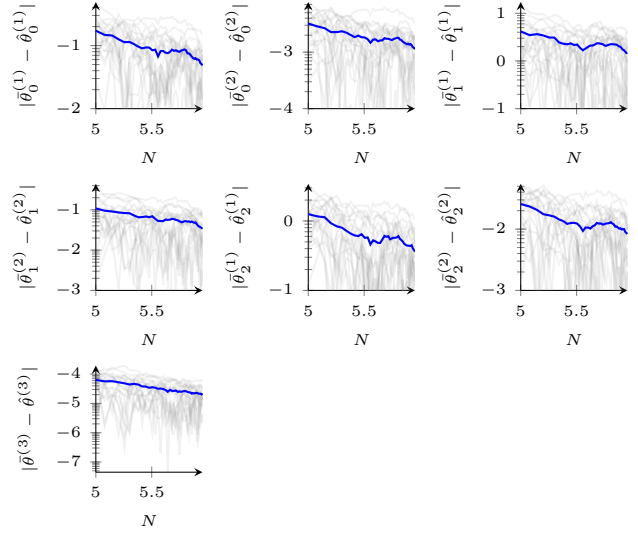
\begin{figure}[t]
	\centering
	\if\usetikz1
	\begin{tikzpicture}
		\begin{groupplot}[
		group style={group size=3 by 3,
			horizontal sep=40pt,
			},
		height=3cm,width=3cm,
		xmode=log,
		ymode=log,
		log base 10 number format code/.code={\pgfmathprintnumber[fixed]{#1}}, 
		axis y line = left,
		xlabel={$N$}, 
		axis x line = bottom,	
		tick label style={font=\tiny},	
		label style={font=\tiny},
		]
		
			\nextgroupplot[ylabel=
		{$|\bar{\theta}_0^{(1)}-\hat{\theta}_0^{(1)}|$},
						ymin=0.01]
			\foreach \m in {3,4,...,21}{
			\addplot [thick,color=gray,opacity=0.1]
				table[x index=0,y index=\m] 
				{./Data/HH_theta1.txt};
				}
			\addplot [thick,color=blue]
				table[x index=0,y index=1] 
				{./Data/HH_theta1.txt};
			\nextgroupplot[ylabel=
		{$|\bar{\theta}_0^{(2)}-\hat{\theta}_0^{(2)}|$},
						ymin=0.0001]
			\foreach \m in {3,4,...,21}{
			\addplot [thick,color=gray,opacity=0.1]
				table[x index=0,y index=\m] 
				{./Data/HH_theta2.txt};
				}
			\addplot [thick,color=blue]
				table[x index=0,y index=1] 
				{./Data/HH_theta2.txt};
			\nextgroupplot[ylabel=
		{$|\bar{\theta}_1^{(1)}-\hat{\theta}_1^{(1)}|$},
						ymin=0.1]
			\foreach \m in {3,4,...,21}{
			\addplot [thick,color=gray,opacity=0.1]
				table[x index=0,y index=\m] 
				{./Data/HH_theta3.txt};
				}
			\addplot [thick,color=blue]
				table[x index=0,y index=1] 
				{./Data/HH_theta3.txt};
			\nextgroupplot[ylabel=
		{$|\bar{\theta}_1^{(2)}-\hat{\theta}_1^{(2)}|$},
							ymin=0.001]
			\foreach \m in {3,4,...,21}{
			\addplot [thick,color=gray,opacity=0.1]
				table[x index=0,y index=\m] 
				{./Data/HH_theta5.txt};
				}
			\addplot [thick,color=blue]
				table[x index=0,y index=1] 
				{./Data/HH_theta5.txt};
			\nextgroupplot[ylabel=
		{$|\bar{\theta}_2^{(1)}-\hat{\theta}_2^{(1)}|$},
						ymin=0.1]
			\foreach \m in {3,4,...,21}{
			\addplot [thick,color=gray,opacity=0.1]
				table[x index=0,y index=\m] 
				{./Data/HH_theta4.txt};
				}
			\addplot [thick,color=blue]
				table[x index=0,y index=1] 
				{./Data/HH_theta4.txt};
			\nextgroupplot[ylabel=
		{$|\bar{\theta}_2^{(2)}-\hat{\theta}_2^{(2)}|$},
						ymin=0.001]
			\foreach \m in {3,4,...,21}{
			\addplot [thick,color=gray,opacity=0.1]
				table[x index=0,y index=\m] 
				{./Data/HH_theta6.txt};
				}
			\addplot [thick,color=blue]
				table[x index=0,y index=1] 
				{./Data/HH_theta6.txt};
			\nextgroupplot[ylabel=
		{$|\bar{\theta}^{(3)}-\hat{\theta}^{(3)}|$}]
			\foreach \m in {3,4,...,21}{
			\addplot [thick,color=gray,opacity=0.1]
				table[x index=0,y index=\m] 
				{./Data/HH_theta7.txt};
				}
			\addplot [thick,color=blue]
				table[x index=0,y index=1] 
				{./Data/HH_theta7.txt};
				
		\end{groupplot}
	\end{tikzpicture}
	\else
		\includegraphics[scale=1]{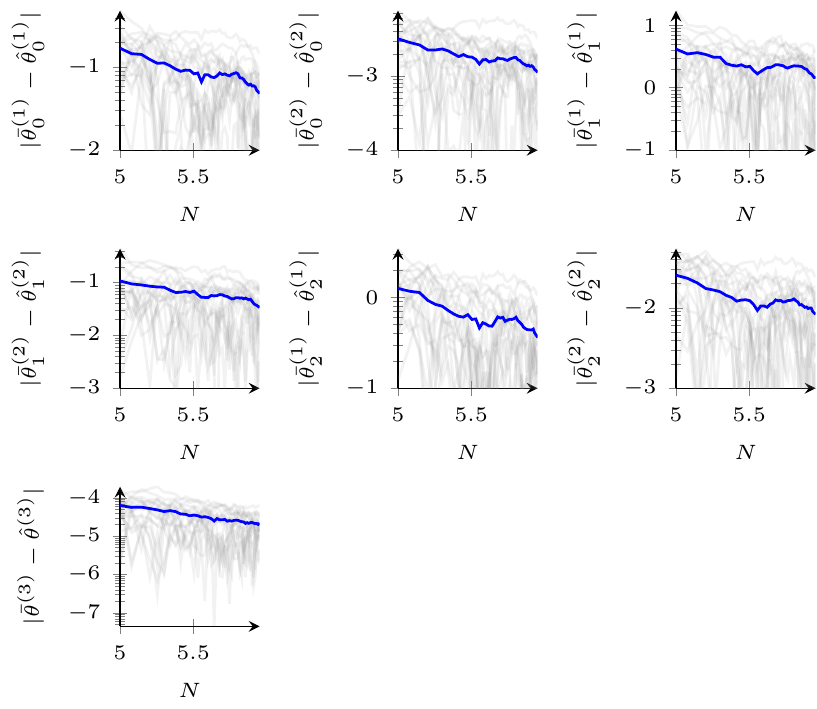}
	\fi
	\caption{The $\log_{10} \times \log_{10}$ plots above
	show how the errors in the estimated parameters of
	Example \ref{ex:id_HH} fall as the number of data
	points $N$ increases. In grey: errors in each of the 
	20 realizations of the identification experiment as
	computed for $N=10^5$ to $9\times 10^5$ ($\samplingT = 
	0.005$). In blue: average of the 20 error traces.}
	\label{fig:HH_estimation}
\end{figure}

\begin{exmp}
\label{ex:CS} 

In this example, we illustrate how a library of pre-established
set of ion channel kinetic models can be used to identify
different neuronal models. We consider three models, all of which
are based on the system given by
\begin{equation*}
	\begin{split}
	\tfrac{v_{k+1}-v_k}{\samplingT} &= 
	-0.3(v_k+17) -\sum_{j=1}^4 i_{j,k} 
	+ \gain(\vref_k-v_k) + e_k \\
		i_{1,k} &= 120 \, m_{1,k}^3 \, h_{1,k}
		(v_k-55)\\
		i_{2,k} &= 20 \, m_{2,k}^4
		(v_k+75)\\
		i_{3,k} &= \gmax_3 \, m_{3,k}^3 \, h_{3,k}
		(v_k+75)\\
		i_{4,k} &= \gmax_4 \, m_{4,k}^2
		(v_k-120)\\
	\end{split}
\end{equation*}
where the states $m_j$ and $h_j$ are given by the forward-Euler
discretization of \eqref{eq:activation} and
\eqref{eq:inactivation}, respectively. The functions 
$m_{\infty,j}$, $h_{\infty,j}$, $\tau_{m,j}$ and $\tau_{h,j}$
are plotted in Figure \ref{fig:CS_kinetic_functions}, and are
described in Appendix \ref{app:cs_kinetics}.

The above system, taken from \cite{drion_ion_2015},
defines a modified version of the Connor-Stevens  
neuronal model \cite{connor_neural_1977}. The values of the
variables $\gmax_3$ and $\gmax_4$ are the distinguishing
factors between the three models we use in
this example. We call them Connor-Stevens (CS) models A, B,
and C, according to the maximal conductance values found
in Table \ref{tab:CS_pars}.

\begin{table}
\centering
\begin{tabular}{|c|c|c|c|}
\hline 
CS model & A & B & C \\ 
\hline 
$\gmax_3$ & 0 & 90 & 0 \\ 
\hline 
$\gmax_4$ & 0 & 0 & 0.4\\ 
\hline 
\end{tabular}
\caption{True maximal conductances in CS models A, B, and C.}
\label{tab:CS_pars}  
\end{table}

Connor Stevens model A is similar to the HH model
of the previous example, while models B and C differ
from A due to the addition of ion currents $i_3$ and $i_4$,
respectively (these currents represent an ``A-type'' 
potassium current and a calcium current, respectively). It
can be verified through
simulations that the addition of $i_3$ or $i_4$ makes the
qualitative input-output behavior (from $\Iapp$ to $v$) of
models B and C differ from that of model A.
In particular,   models B and C can   fire 
periodic spikes with arbitrarily low frequency, while model
A does not have that property (see, for instance, Figure 2
of \cite{drion_ion_2015}). The property of spiking with
arbitrarily low frequency has important neurocomputational
consequences. It underlies the classical distinction between Type I and Type II  
neuronal excitability first proposed by Hodgkin and
Huxley (see \cite[Chapter 7]{izhikevich_dynamical_2007}).

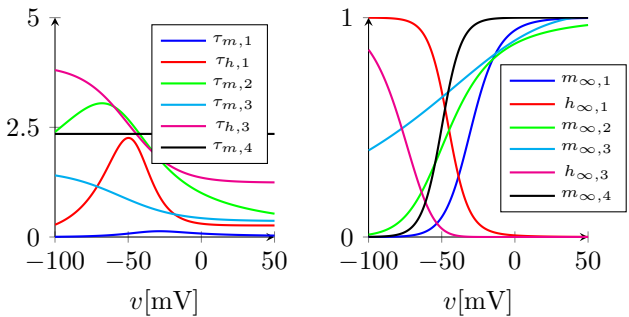
\begin{figure}[b]
	\centering
	\if\usetikz1
	\begin{tikzpicture}
		\begin{axis}[height=4.5cm,
		width=4.5cm,
		axis y line = left,
		xlabel={$v \mathrm{[mV]}$}, axis x line = bottom,
		ymin=0,ymax=5,
		ytick={0,2.5,5},
		legend pos= north east,
		name = CS_tau
		]
			\addplot[smooth,thick,color=blue] 
				file{./Data/tau_m_cs.txt};
			\addlegendentry{\tiny $\tau_{m,1}$};
			\addplot[smooth,thick,color=red] 
				file{./Data/tau_h_cs.txt};
			\addlegendentry{\tiny $\tau_{h,1}$};
			\addplot[smooth,thick,color=green] 
				file{./Data/tau_n_cs.txt};
			\addlegendentry{\tiny $\tau_{m,2}$};
			\addplot[smooth,thick,color=cyan] 
				file{./Data/tau_a_cs.txt};
			\addlegendentry{\tiny $\tau_{m,3}$};
			\addplot[smooth,thick,color=magenta] 
				file{./Data/tau_b_cs.txt};
			\addlegendentry{\tiny $\tau_{h,3}$};
			\addplot[smooth,thick,color=black] 
				file{./Data/tau_c_cs.txt};
			\addlegendentry{\tiny $\tau_{m,4}$};
		\end{axis}
		\begin{axis}[height=4.5cm,
		width=4.5cm,
		axis y line = left,
		xlabel={$v \mathrm{[mV]}$}, axis x line = bottom,
		ymin=0,ymax=1,
		legend style={at={(0.6,0.45)},anchor=west},
		ytick={0,1},
		at={(CS_tau.north east)},
		anchor = north west,
		xshift = 1.25cm
		]
			\addplot[smooth,thick,color=blue] 
				file{./Data/sig_m_cs.txt};
				\addlegendentry{\tiny $m_{\infty,1}$}; 
			\addplot[smooth,thick,color=red] 
				file{./Data/sig_h_cs.txt};
				\addlegendentry{\tiny $h_{\infty,1}$};
			\addplot[smooth,thick,color=green] 
				file{./Data/sig_n_cs.txt};
				\addlegendentry{\tiny $m_{\infty,2}$};
			\addplot[smooth,thick,color=cyan] 
				file{./Data/sig_a_cs.txt};
				\addlegendentry{\tiny $m_{\infty,3}$};
			\addplot[smooth,thick,color=magenta] 
				file{./Data/sig_b_cs.txt};
				\addlegendentry{\tiny $h_{\infty,3}$};
			\addplot[smooth,thick,color=black] 
				file{./Data/sig_c_cs.txt};
				\addlegendentry{\tiny $m_{\infty,4}$};
		\end{axis}
	\end{tikzpicture}
	\else
		\includegraphics[scale=1]{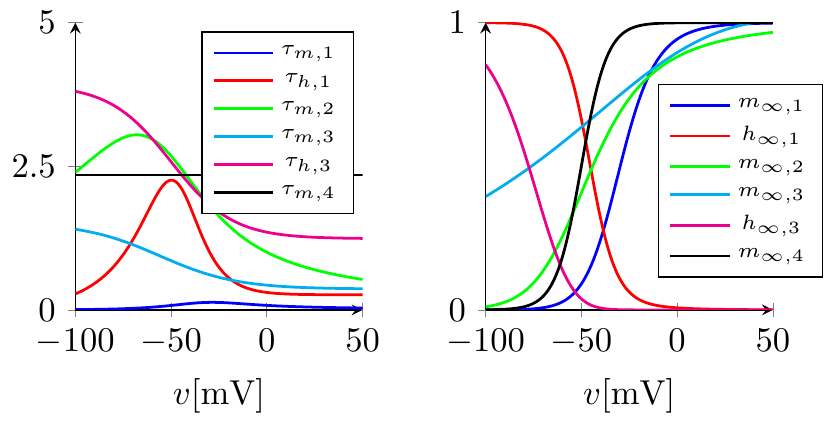}
	\fi
	\caption{Left: time constant functions 
	$\tau_{m,j}$ and $\tau_{h,j}$ in the Connor-Stevens
	model. Right: nonlinear activation functions
	$m_{\infty,j}$ and $h_{\infty,j}$ in the Connor-Stevens
	model.}
	\label{fig:CS_kinetic_functions}
\end{figure}

To identify the models A, B and C, we include
in a single model structure all four of the ion channels
shared by those models. 
We simulated identification experiments in which $\gain = 50$
and $\vref_k = -45 + \tilde{\vref}_k$, where $\tilde{\vref}_k$ is
white Gaussian noise of standard deviation $\sigma_r = 30 \, \mathrm{mV}$
that is first filtered by the zero-order hold discretization of the 
system $10^2/(s+10)^2$, then truncated so that $|\tilde{\vref}_k| \le 30$
for $k\ge 0$. The input noise $e_k$ is white Gaussian noise of 
$\sigma_e = 1 \, \mathrm{\upmu A/cm^2}$ that is truncated so that 
$|e_k| \le 20$ for $k\ge 0$. This setup resulted in a signal-to-noise ratio
(between $y_k$ and $e_k$) of around $28$ dB, $26$ dB and $29$ dB for the CS
models A, B and C, respectively (again,we eliminated the first $0.5$ seconds
of measurement from all datasets).

\begin{figure}[ht]
	\centering
	\if\usetikz1
	\begin{tikzpicture}
		\begin{groupplot}[
		group style={group size=2 by 2,
			horizontal sep=40pt,
			},
		height=4.0cm,width=4.0cm,
		axis y line = left,
		xlabel={$N$}, 
		axis x line = bottom,	
		tick label style={font=\scriptsize},	
		label style={font=\scriptsize},
		]
			\nextgroupplot[ylabel={$\gmax_1$},ymin=118,
			ymax=122]
			\addplot [thick,color=blue]
				table[x index=0,y index=1] 
				{./Data/idconds_drion2015a.txt};
			\addplot [thick,color=red]
				table[x index=0,y index=1] 
				{./Data/idconds_drion2015b.txt};
			\addplot [thick,color=green]
				table[x index=0,y index=1] 
				{./Data/idconds_drion2015e.txt};
			\addplot[dashed,black,domain=20000:500000]
					{120};
			\nextgroupplot[ylabel=
						{$\gmax_2$},ymin=18]%
			\addplot [thick,color=blue]
				table[x index=0,y index=2] 
				{./Data/idconds_drion2015a.txt};
			\addplot [thick,color=red]
				table[x index=0,y index=2] 
				{./Data/idconds_drion2015b.txt};
			\addplot [thick,color=green]
				table[x index=0,y index=2] 
				{./Data/idconds_drion2015e.txt};
			\addplot[dashed,black,domain=20000:500000]{20};
			\nextgroupplot[ylabel=
						{$\gmax_3$},ymax=100]
			\addplot [thick,color=blue]
				table[x index=0,y index=3] 
				{./Data/idconds_drion2015a.txt};
			\addplot [thick,color=red]
				table[x index=0,y index=3] 
				{./Data/idconds_drion2015b.txt};
			\addplot [thick,color=green]
				table[x index=0,y index=3] 
				{./Data/idconds_drion2015e.txt};
			\addplot[dashed,black,domain=20000:500000]{90};
			\addplot[dashed,black,domain=20000:500000]{0};
			\nextgroupplot[ylabel=
						{$\gmax_4$},ymax=0.5]%
			\addplot [thick,color=blue]
				table[x index=0,y index=4] 
				{./Data/idconds_drion2015a.txt};
			\addplot [thick,color=red]
				table[x index=0,y index=4] 
				{./Data/idconds_drion2015b.txt};
			\addplot [thick,color=green]
				table[x index=0,y index=4] 
				{./Data/idconds_drion2015e.txt};
			\addplot[dashed,black,domain=20000:500000]
				{0.4};
			\addplot[dashed,black,domain=20000:500000]{0};

		\end{groupplot}
	\end{tikzpicture}
	\else
		\includegraphics[scale=1]{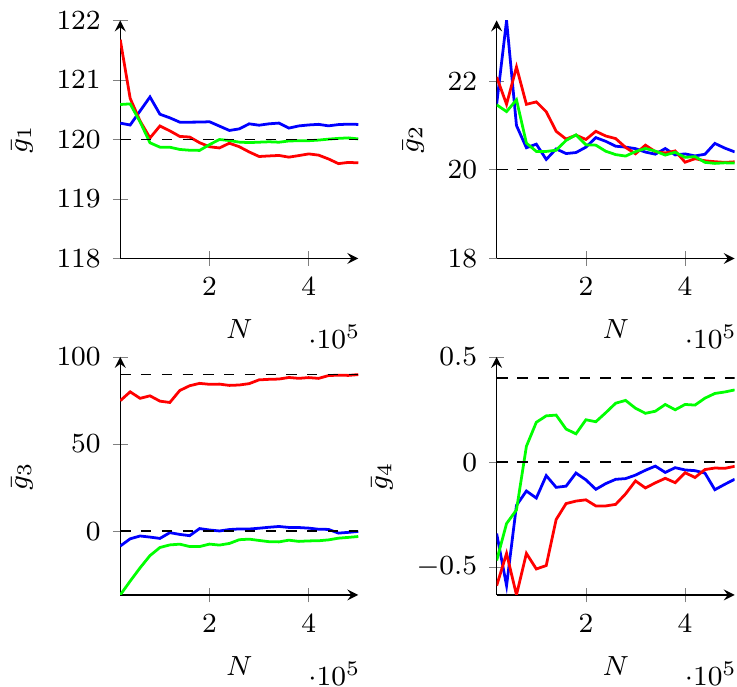}
	\fi
	\caption{Evolution of the estimates of $\gmax_j$,
	with respect to the number of samples,
	for each of the identified Connor-Stevens models A
	(blue),	B (red) and C (green). The sampling period is 
	$\samplingT = 0.005$, and the experimental setup
	is described in Example \ref{ex:CS}.}
	\label{fig:CS_estimation}
\end{figure}
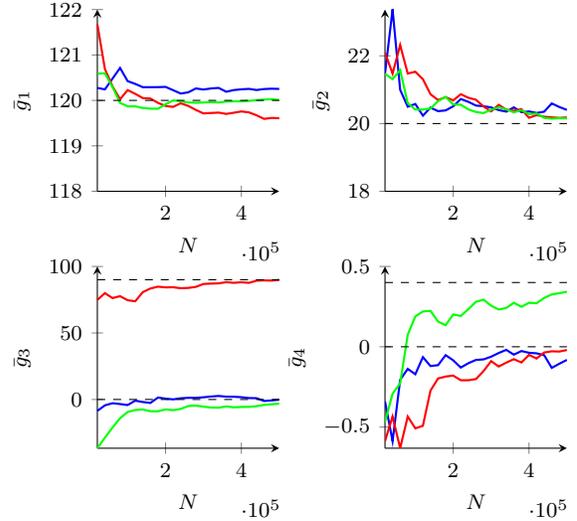

Figure \ref{fig:CS_estimation} shows the evolution of
the estimates of $\gmax_j$ obtained by identifying each
of the CS models A, B and C (for brevity, we do not show
the evolution of all parameter estimates). It can be seen
that the estimates of $\gmax_3$ (or $\gmax_4$) for models
that do not contain $i_3$ (or $i_4$) tend towards zero,
while the other estimates tend towards their true values.

\end{exmp}

\section{Conclusion}
\label{sec:conclusion} 

In this paper, we studied the identification of 
discrete-time neuronal systems under the assumption of 
current-additive zero-mean white noise and negligible 
voltage measurement noise. We showed that by treating 
a neuronal model as a closed-loop system, we can solve 
the identification problem by identifying the inverse 
dynamics with an output-error model structure. We have 
demonstrated that consistent parameter estimates are 
obtained when the model structure contains the internal
dynamics of the system being identified. This is a common
strategy adopted in neuroscience, where kinetic models
of ion channels are estimated in separate experiments
(see, e.g., \cite{mcdougal_twenty_2017}).
It is worth noting that the results in this paper may hold
for ion channel models which are more general than
\eqref{eq:activation}-\eqref{eq:inactivation};
the key requirement is that the ion channels
possess a contracting dynamics, so that
\eqref{eq:contracting_internal} is satisfied.
Thus, this work rigorously justifies neuronal system
identification using conventional methods of nonlinear 
identification.

\begin{ack}                               
Thiago Burghi was supported by the 
\textit{Coordena\c{c}\~{a}o de Aperfei\c{c}oamento 
de Pessoal de N\'{i}vel Superior} (CAPES) -- Brasil 
(Finance Code 001).
Maarten Schoukens was supported by the 
European Union's Horizon 2020 research and innovation 
programme under the Marie Sklodowska-Curie Fellowship 
(grant agreement nr. 798627). The research leading to 
these results has received funding from the European 
Research Council under the Advanced ERC Grant Agreement 
Switchlet n.670645. The authors thank the anonymous 
reviewers, as well as Dr. Monika Josza, for helping to
improve earlier versions of this manuscript.
\end{ack}

\appendix
\section{Proofs}
\subsection{Proof of Lemma \ref{lem:contraction_expo_stab}}
\label{proof:contraction_expo_stab} 

Let $P = \Theta^\top\Theta$, where $\Theta > 0$. Applying the
change of coordinates $z_k = \Theta x_k$, we obtain the 
discrete-time dynamics
\begin{equation}
	\label{eq:z_dynamics}
	z_{k+1} = f_\Theta(z_k,u_k),
\end{equation}
where $f_\Theta$ is given by 
\begin{equation}
	\label{eq:zeta_dynamics} 	
	f_\Theta(\zeta,\upsilon) = \Theta 
	f\left(\Theta^{-1}\zeta,\upsilon\right).
\end{equation}
By the assumptions on $X$, the set
\[	
	Z = \{\zeta \in \setreal^\nstates \;|
	\;\zeta=\Theta\xi,\;\xi \in X\}.
\]
is closed, bounded and convex. Furthermore, $Z$ is a positively
invariant set for \eqref{eq:z_dynamics}, uniformly in 
$[-\beta,\beta]^\ninputs$.

Since $P = \Theta^\top\Theta$ with $\Theta$ invertible, the
inequality \eqref{eq:DT_contraction} implies
\begin{equation*}
	\sigma_{\max}\left[\Theta
	\frac{\partial f}{\partial x}(x,u)
		\Theta^{-1} \right] \le \alpha < 1
\end{equation*}
for all $k \in \setint_+$, $x \in X$, and $u \in U$.
From \eqref{eq:zeta_dynamics}, this implies that 
$\sigma_{\max}\left[\partial f_\Theta/
		\partial\zeta\right]\le \alpha$ on
$Z\times[-\beta,\beta]^\ninputs$. 
Furthermore, since $\partial f_\Theta/ \partial \upsilon$
is a continuous function and
$Z~\times~[-\inpbound,\inpbound]^\ninputs$ is closed
and bounded, there is some $L_1>0$ such that 
$\sigma_{\max}\left[\partial f_\Theta/
		\partial \upsilon\right]\le L_1$
on $Z \times [-\inpbound,\inpbound]^\ninputs$.

Now, let $\zeta,\tilde{\zeta} \in Z$ and
$\upsilon,\tilde{\upsilon} \in 
[-\inpbound,\inpbound]^\ninputs$. Let also 
$\gamma_1(s)~=~(1-s)\tilde{\zeta} + s \zeta$ and
$\gamma_2(s) = (1-s)\tilde{\upsilon} + s \upsilon$,
with $s \in [0,1]$. It can be shown, using the mean
value theorem (see, e.g., the proof of 
\cite[Lemma 3.1]{khalil_nonlinear_2002}), that 
there is an $s^* \in (0,1)$ such that 
\begin{equation*}
\resizebox{\linewidth}{!}{$
	\begin{aligned}
	\|f_\Theta(\zeta,\upsilon)-
		f_\Theta(\tilde{\zeta},\tilde{\upsilon})\|
	\le 
	&\left\| \frac{\partial f_\Theta}{\partial \zeta}
		(\gamma_1(s^*),\gamma_2(s^*))
	(\zeta-\tilde{\zeta})\right. 
	\\
	&\quad+ 
	\left.\frac{\partial f_\Theta}{\partial \upsilon}
		(\gamma_1(s^*),\gamma_2(s^*))
	(\upsilon-\tilde{\upsilon})	\right\|
	\end{aligned}
$}
\end{equation*}

By the triangle inequality and convexity of 
$Z \times [-\inpbound,\inpbound]^\ninputs$, the above implies
\begin{equation}
	\label{eq:lipschitz_inequality} 
	\|f_\Theta(\zeta,\upsilon) - f_\Theta(\tilde{\zeta},
	\tilde{\upsilon})\|  
	\le \alpha \|\zeta-\tilde{\zeta}\| + 
	L_1 \|\upsilon-\tilde{\upsilon}\|
\end{equation}
on $Z \times [-\inpbound,\inpbound]^\ninputs$.
By positive invariance of $Z$, we are allowed to apply 
\eqref{eq:z_dynamics} and \eqref{eq:lipschitz_inequality}
recursively, obtaining
\begin{equation*} 
\label{eq:expo_stable_z} 
\resizebox{\linewidth}{!}{$
	\begin{aligned}
		\|z_k - \tilde{z}_k\|
		&\le L_1
		\sum_{m=1}^{k} \alpha^{m-1}
		\|u_{k-m} - \tilde{u}_{k-m}\|
		+  \alpha^k \|z_0 - \tilde{z}_0\|
	\end{aligned}
$}
\end{equation*}
for $k\ge 0$. Multiplying both sides of the
inequality by $\sigma_{\max}\left[\Theta^{-1}\right]$ and 
substituting $z_k = \Theta x_k$, we have
\begin{equation}
	\label{eq:expo_stable_states} 
	\begin{split}
		\|x_k - \tilde{x}_k\|
		\le& \frac{L_1}{\sigma_{\min}}
		\sum_{m=1}^{k} \alpha^{m-1}
		\|u_{k-m} - \tilde{u}_{k-m}\|\\
		&+ \frac{\sigma_{\max}}{\sigma_{\min}}\;
		\alpha^k \;\|x_0 - \tilde{x}_0\|
	\end{split}
\end{equation}
for $k\ge 0$, where $\sigma_{\max}$ and $\sigma_{\min}$
denote the largest and the smallest singular values of 
$\Theta$, respectively.

By arguments similar to those used above, 
there are $L_2,L_3>0$ such that 
\begin{equation}
	\label{eq:lipschitz_output} 
	\|y_k - \tilde{y}_k\| \le 
	L_2 \|x_k - \tilde{x}_k\|  + L_3 \|u_k-\tilde{u}_k\|
\end{equation}
The result \eqref{eq:expo_stable_output} follows 
directly from \eqref{eq:expo_stable_states} and 
\eqref{eq:lipschitz_output} by setting
$C_1 = \max\{L_1L_2\sigma_{\min}^{-1},L_3\}$
and $C_2 = L_2 \sigma_{\max}\sigma_{\min}^{-1}$.

\subsection{Proof of Proposition \ref{prop:DT_contraction}}
\label{proof:DT_contraction}

To prove Proposition \ref{prop:DT_contraction}, we first
state a result concerning the contraction of forward-Euler
discretized systems:
\begin{lem}
	\label{lem:discretized_contracting} 
	Consider the continuous-time dynamics
	\begin{equation}
		\label{eq:aux_CT_system} 
		\dot{x}(t) = f(x(t),u(t)) + B_r \,r(t),
	\end{equation}
	where $x(t)\in\setreal^{\nstates}$, 
	$u(t)\in \setreal^{\ninputs}$, 
	$r(t)\in\setreal^{\nref}$,
    $B_r$ is a constant matrix, and $f$ is 
    continuously differentiable. 
    Assume \eqref{eq:aux_CT_system} is exponentially 
    contracting in a set $X$, uniformly in $(u,r)$ on 
    $\setreal^{\ninputs+\nref}$, with constant  
    $P>0$ and $\lambda>0$. Assume $\partial f/\partial x$ is bounded 
    on $X \times \setreal^{\ninputs}$. 
	Let
	\begin{equation}
		\label{eq:aux_DT_system}
		\begin{split}
			x_{k+1} &= f_d(x_k,u_k,d_k) \\
			&:= x_k + \samplingT (f(x_k,u_k) + B_d \, d_k) 
		\end{split}
	\end{equation}
	where $\samplingT>0$ is a sampling period, $d_k\in\setreal^{n_d}$, 
	and $B_d$ is a constant matrix. Then, there exists a sufficiently 
	small $\samplingT$ such that \eqref{eq:aux_DT_system} is 
	exponentially contracting in $X$, uniformly in $(u,d)$ on
	$\setreal^{\ninputs+n_d}$.
\end{lem}
\begin{proof}
Since $\partial f/\partial x$ is bounded on 
$X \times \setreal^{\ninputs}$, there is a number $\bar{\sigma}$ 
such that $\bar{\sigma}\ge\sigma_{\max}[\partial f/\partial x]$
on that set. Using contraction of the continuous-time system, 
we have
\begin{equation}
	\label{eq:proof_discretized_contraction} 
	\begin{split}
		\frac{\partial f_d^\top}{\partial x }P
		\frac{\partial f_d}{\partial x } 
		&= \left(I + \samplingT \frac{\partial f^\top}
							{\partial x}\right)
		P \left(I + \samplingT \frac{\partial f}
							{\partial x}\right)\\
		&\le (1 - 2\samplingT\lambda) P + \samplingT^2 
		\frac{\partial f^\top}{\partial x}P
		\frac{\partial f}{\partial x} \\
		&\le \left(1-2\samplingT\lambda + \samplingT^2 
		\frac{\lambda_{\max}[P]}{\lambda_{\min}[P]}
		\bar{\sigma}^2
			\right) P \\
			&= \alpha(\samplingT)^2 P
	\end{split}
\end{equation}
for all $x \in X$ and $u\in\setreal^{\ninputs}$. The second 
inequality above follows from the fact that 
$A^\top P A \le \lambda_{\max}[P] \sigma^2_{\max}[A] I$ 
and $I \le 1/\lambda_{\min}[P] P$.
Making $\samplingT<1$ small enough ensures that 
$\alpha(\samplingT)^2 < 1$, concluding the proof.
\end{proof}

We now carry on with the proof Proposition 
\ref{prop:DT_contraction}. The discretization of
\eqref{eq:activation} is given by
\begin{equation*}
\begin{split}
	m_{j,k+1} &= m_{j,k} + \tfrac{\samplingT}{\tau_{m,j}(v_k)}
	(-m_{j,k} + m_{\infty,j}(v_k))\\
			&= m_{j,k}
			\left(1-\tfrac{\samplingT}{\tau_{m,j}(v_k)}\right)
			+ \tfrac{\samplingT}{\tau_{m,j}(v_k)}
			m_{\infty,j}(v_k)
\end{split}
\end{equation*}
where $m_{\infty,j}(v_k)\in[0,1]$ and 
$\tau_{m,j}(v_k)\in[\tau_{\min},\tau_{\max}]$, with 
$\tau_{\min}~>~0$. It directly follows that for any
$\samplingT \le \tau_{\min}$, for all $m_{j,k}\in[0,1]$, and for
all $v_k\in\setreal$, we have $m_{j,k+1}\in~[0,1]$. An analogous
fact holds for $h_{j,k}$. Thus for any 
$\samplingT \le \tau_{\min}$, $w_0 \in~[0,1]^\ngat$ implies 
$\w_k \in [0,1]^\ngat$ for all $k\ge 0$, and $[0,1]^{\ngat}$
is positively invariant for the subsystem 
\eqref{eq:DT_int_dynamics}, uniformly in $v$ on $\setreal$.
Since $A(v)=-\mathrm{diag}(1/\tau_{m,1}(v),\dotsc)$ is bounded 
on $\setreal$, Proposition \ref{prop:contractive} together 
with Lemma \ref{lem:discretized_contracting} imply the
existence of a $t_s \le \tau_{\min}$ such that 
\eqref{eq:DT_int_dynamics} is exponentially contracting
in $[0,1]^{\ngat}$, uniformly on $\setreal$.

Now, let $V_{\gamma}\subset \setreal$ be the interval
	\begin{equation*}
		\label{eq:V_domain} 
		V_{\gamma} = \left[
		\min_j \{\nu_j, 
		-\beta\tfrac{\gain+1}{\gain}\},
		\max_j \{\nu_j, 
		\beta\tfrac{\gain+1}{\gain}\}
		\right]
	\end{equation*}
Let $\samplingT^*(\gamma) = 
		\min\{\tau_{\min},1/\gamma\}$, and let
\[
	v_{\max}(\gamma) = 
	\max_{\substack{v_k,\w_k,\\\vref_k,e_k}}
	 v_k - 
	 \frac{\samplingT^*(\gamma)}{c}(\Iint(v_k,\w_k)- 
	 \gamma(\vref_k-v_k)-e_k)
\]
where the maximum is over the closed and bounded set
$V_{\gamma}\times[0,1]^{\ngat} \times[-\beta,\beta]^2$.
Defining $v_{\min}(\gamma)$ analogously,
we claim that for
all $\gamma>0$, $[v_{\min}(\gamma),v_{\max}(\gamma)]
\times[0,1]^\ngat$ is a positively invariant set for
\eqref{eq:DT_CL_model}, uniformly on $[-\beta,\beta]^2$.
To prove this claim, we first observe that 
$v_{k+1} \in [v_{\min}(\gamma),v_{\max}(\gamma)]$ 
whenever 
$(v_{k},\w_{k}^\top)\in V_{\gamma}\times[0,1]^{\ngat}$
and $(\vref_k,e_k)\in [-\beta,\beta]^2$. 
If $[v_{\min}(\gamma),v_{\max}(\gamma)] \subseteq V_\gamma$, then
the claim follows immediately from the previous observation. If,
alternatively, 
$V_\gamma \subset [v_{\min}(\gamma),v_{\max}(\gamma)] $, then the
claim follows from the fact that 
\[ 
\begin{array}{ll}
v_{k+1} \ge v_k \quad &\text{for all} \quad 
v_k \le \min_j\{\nu_j,-\beta(\gamma+1)/\gamma\}
\\
v_{k+1} \le v_k \quad &\text{for all} \quad 
v_k \ge \max_j\{\nu_j,\beta(\gamma+1)/\gamma\} 
\\
\end{array}
\]
whenever $\w_k\in[0,1]^\ngat$ and $(\vref_k,e_k)\in 
[-\beta,\beta]^2$. 

Since $[v_{\min}(\gamma),v_{\max}(\gamma)]$ is uniformly
bounded in $\gamma>0$, Proposition \ref{prop:output_contractive} 
together with Lemma \ref{lem:discretized_contracting} imply 
the existence of a $\gamma>0$ and a $t_s \le t_s^*(\gamma)$ 
such that \eqref{eq:DT_CL_model} is exponentially contracting in 
$[v_{\min}(\gamma),v_{\max}(\gamma)]\times [0,1]^\ngat$, 
uniformly in $[-\beta,\beta]^\ngat$. This concludes the proof.

\subsection{Proof of Lemma \ref{lem:S3}}
\label{proof:S3} 

Consider two different solutions of 
\eqref{eq:true_feedback}-\eqref{eq:true_system} (which, combined, 
can be written as \eqref{eq:DT_CL_model}). The first is given by 
\begin{equation}
	\label{eq:sol1}
	\begin{split}
		(v_k,\w_k^\top)^\top &= 
		\phi_{k,0}((\vref,e)^\top,(v_0,\w_0^\top)^\top) \\
	\end{split} 
\end{equation}
for $k\ge 0$, and the second is given by
\begin{equation}
	\label{eq:sol2}
	\begin{split}
(\bar{v}_{k,s+1},\bar{\w}_{k,s+1}^\top)^\top &= 
\phi_{k,s+1}((\vref,e)^\top,(\bar{v}_{s+1},
							\bar{\w}_{s+1}^\top)^\top) \\
	(\bar{v}_{s+1},\bar{\w}_{s+1}^\top)^\top &= 0
	\end{split} 
\end{equation}
for $k \ge s + 1$. 

We will use the solutions above to construct the random
variables $\bar{y}_{k,s}$ involved in Condition
\ref{cond:S3}. First, for each $s\in\setint_+$, we set 
$\bar{y}_{s,s} = 0$.
From \eqref{eq:true_system_y}, we compute the sequence 
$y_k$ using \eqref{eq:sol1}, for $k\ge 0$, and the sequence
$\bar{y}_{k,s+1}$ using \eqref{eq:sol2}, for $k \ge s+1$.
We have that $\bar{y}_{k,s}$ is independent of 
$e_{[0,s]}$, since $e_{[s+1,k]}$ is independent of
$e_{[0,s]}$; furthermore, $\vref_k$ is deterministic;
thus the independence required in Condition
\ref{cond:S3} is satisfied.
We now need to verify \eqref{eq:condition_S3_A}
for $k \ge s$. For $k=s$, we have
\begin{equation}
	\label{eq:bound_S3_s0} 
	\begin{split}
	|y_s - \bar{y}_{s,s}| &= |y_s| = 
	\tfrac{1}{c} \big| \Iint(v_s,\w_s) - \gamma(r_s-v_s)
		-e_s \big|\\
		&\le \tfrac{1}{c}\big(|\Iint(v_s,\w_s)+\gamma v_s|
			+ (\gamma + 1) \inpbound  \big) \\
		&\le C_1 
	\end{split}
\end{equation}
for some $C_1>0$ and for each $s\in\setint_+$. To ensure
this bound, we have used (from Assumptions
\ref{asp:discrete_input}-\ref{asp:contracting_true_system}) the
fact that $(\vref,e) \in \inpclass^2$, and the fact that 
$\Iint(v,\w)+\gamma v$ is a continuous function on the set 
$[v_{\min},v_{\max}]	\times[0,1]^{\ngat}$.

Now, we make use of Assumption
\ref{asp:contracting_true_system}. Let $\alpha_c<1$ be the
contraction rate of the closed-loop dynamics 
\eqref{eq:true_feedback}-\eqref{eq:true_system}.
Since $(\vref,e) \in \inpclass^2$, we can apply 
Lemma~\ref{lem:contraction_expo_stab} (with the time origin
shifted to $s+1$) to see that there is a $C_2 > 0$ such that
\begin{equation}
	\label{eq:bound_S3}
	\resizebox{\linewidth}{!}{$ 
	\begin{aligned}
	\left|y_k - \bar{y}_{k,s+1}\right|  
		&\le C_2 \alpha_c^{k-(s+1)} \,
		\|(v_{s+1},\w_{s+1}^\top) - 
			(\bar{v}_{s+1},\bar{\w}_{s+1}^\top)\|\\
			&= \alpha_c^{-1} C_2 \, \alpha_c^{k-s} \,
			\|(v_{s+1},\w_{s+1}^\top)\| \\
			&\le \alpha_c^{-1} C_2 \, C_3 \,\alpha_c^{k-s}
	\end{aligned}
	$}
\end{equation}
for each $s \in \setint_+$ and $k\ge s+1$, where
the constant $C_3>0$ comes from the boundedness of 
$[v_{\min},v_{\max}]	\times[0,1]^{\ngat}$. Taking 
$E[\;\cdot\,^4]$ on both sides of \eqref{eq:bound_S3_s0} and
\eqref{eq:bound_S3}, we verify \eqref{eq:condition_S3_A}
with $C=\max\{C_1^4,(\alpha_c^{-1}C_2C_3)^4\}$ and 
$\alpha = \alpha_c^4$. 

The random variables $\bar{u}_{k,s}$ of Condition
\ref{cond:S3} can be constructed in a completely analogous
way, and thus we omit this part of the proof.

\section{Hodgkin-Huxley kinetic functions}
\label{app:kinetics}

To define the ion channel kinetics of the Hodgkin-Huxley
model, we first set
\begin{equation*}
\resizebox{\linewidth}{!}{$
{\renewcommand{\arraystretch}{2}
\begin{array}{ll}
\alpha_{m,1}(v) = 0.1\displaystyle\frac{-40-v}{\,\text{exp}\left( \frac{-40-v}{10}\right)-1} & 
\beta_{m,1}(v) = 4\,\text{exp}\left(\frac{-v-65}{18} \right) \\
\alpha_{h,1}(v)= 0.07\,\text{exp}\left(\frac{-v-65}{20} \right) & 
\beta_{h,1}(v) = \displaystyle\frac{1}{\,\text{exp}\left( \frac{-35-v}{10}\right)+1} \\ 
\alpha_{m,2}(v) = 
0.01\displaystyle\frac{-55-v}{\,\text{exp}\left( \frac{-55-v}{10}\right)-1} & 
\beta_{m,2}(v) = 0.125\,\text{exp}\left(\frac{-v-65}{80} \right)
\end{array} 
}
$}
\end{equation*}

Then, the functions $\tau_{m,j}$ and $m_{\infty,j}$,
$j=1,2$, are given by
\begin{equation}
\label{eq:tau_from_alpha} 
	\begin{split}
	\tau_{m,j}(v) &= \frac{1}{\alpha_{m,j}(v)+
	\beta_{m,j}(v)} \\
	m_{\infty,j}(v) &= 
	\frac{\alpha_{m,j}(v)}{\alpha_{m,j}(v)+
	\beta_{m,j}(v)}
	\end{split}
\end{equation}
The same relationships are used to define 
$\tau_{h,1}$ and $h_{\infty,1}$.

\section{Connor-Stevens kinetic functions}
\label{app:cs_kinetics}

The ion channel kinetics of the CS models are given by
the relationships \eqref{eq:tau_from_alpha}, with
\begin{equation*}
\resizebox{\linewidth}{!}{$
{\renewcommand{\arraystretch}{2}
\begin{array}{ll}
\alpha_{m,1}(v)=0.38\displaystyle\frac{-29.7-v}{\,\text{exp}\left( \frac{-29.7-v}{10}\right)-1} & 
\beta_{m,1}(v)= 15.2\,\text{exp}\left(\frac{-54.7-v}{18} \right) \\ 
\alpha_{h,1}(v)= 0.266\,\text{exp}\left(\frac{-v-48}{20} \right) & 
\beta_{h,1}(v)=3.8\displaystyle\frac{1}{\,\text{exp}\left( \frac{-18-v}{10}\right)+1} \\
\alpha_{m,2}(v)=0.019\displaystyle\frac{-45.7-v}{\,\text{exp}\left( \frac{-45.7-v}{10}\right)-1} 
& \beta_{m,2}(v)=0.2375\,\text{exp}\left(\frac{-55.7-v}{80} \right) \\
\end{array} 
}
$}
\end{equation*}

The remaining functions are given by
\begin{equation*}
\begin{split}
\tau_{m,3}(v) &= 
0.3632+\displaystyle\frac{1.158}{1+\text{exp}\left(\frac{v+55.96}{20.12}\right)}
\\
m_{\infty,3}(v) &= 
\left(0.0761\displaystyle\frac{\text{exp}\left(\frac{v+94.22}{31.84}\right)}
{1+\text{exp}\left(\frac{v+1.17}{28.93}\right)}\right)^{\frac{1}{3}}
 \\
\tau_{h,3}(v) &= 
1.24+\displaystyle\frac{2.678}{1+\text{exp}\left(\frac{v+50}{16.027}\right)}
\\ h_{\infty,3}(v) &= 
 \displaystyle\frac{1}{\left(1+\text{exp}\left(\frac{v+53.3}{14.54}\right)\right)^4}
\end{split}
\end{equation*}
and 
\begin{equation*}
\begin{split}
\tau_{m,4}(v) &= 2.35
\\
m_{\infty,4}(v) &= 
\frac{1}{1+\exp(-0.15(v+50))}
\end{split}
\end{equation*}

\bibliographystyle{plain}
\bibliography{bibliography}

\begin{thebibliography}{10}

\bibitem{abdalmoaty_linear_2019}
Mohamed Rasheed-Hilmy Abdalmoaty and H{\r a}kan Hjalmarsson.
\newblock Linear prediction error methods for stochastic nonlinear models.
\newblock {\em Automatica}, 105:49--63, July 2019.

\bibitem{almog_is_2016}
Mara Almog and Alon Korngreen.
\newblock Is realistic neuronal modeling realistic?
\newblock {\em Journal of Neurophysiology}, 116(5):2180--2209, November 2016.

\bibitem{baladron_mean-field_2012}
Javier Baladron, Diego Fasoli, Olivier Faugeras, and Jonathan Touboul.
\newblock Mean-field description and propagation of chaos in networks of
  {Hodgkin}-{Huxley} and {FitzHugh}-{Nagumo} neurons.
\newblock {\em The Journal of Mathematical Neuroscience}, 2(1):10, May 2012.

\bibitem{boyd_fading_1985}
S.~Boyd and L.~Chua.
\newblock Fading memory and the problem of approximating nonlinear operators
  with {Volterra} series.
\newblock {\em IEEE Transactions on Circuits and Systems}, 32(11):1150--1161,
  November 1985.

\bibitem{burghi_feedback_2019}
Thiago~B. Burghi, Maarten Schoukens, and Rodolphe Sepulchre.
\newblock Feedback for nonlinear system identification.
\newblock In {\em 2019 18th {European} {Control} {Conference} ({ECC})}, pages
  1344--1349, Naples, Italy, June 2019.

\bibitem{byrnes_passivity_1991}
C.~I. Byrnes, A.~Isidori, and J.~C. Willems.
\newblock Passivity, feedback equivalence, and the global stabilization of
  minimum phase nonlinear systems.
\newblock {\em IEEE Transactions on Automatic Control}, 36(11):1228--1240,
  November 1991.

\bibitem{casas_prediction_2002}
Ra{\'u}l~A. Casas, Robert~R. Bitmead, Clas~A. Jacobson, and C.~Richard Johnson.
\newblock Prediction error methods for limit cycle data.
\newblock {\em Automatica}, 38(10):1753--1760, October 2002.

\bibitem{connor_neural_1977}
J~A Connor, D~Walter, and R~McKown.
\newblock Neural repetitive firing: modifications of the {Hodgkin}-{Huxley}
  axon suggested by experimental results from crustacean axons.
\newblock {\em Biophysical Journal}, 18(1):81--102, April 1977.

\bibitem{drion_ion_2015}
Guillaume Drion, Timothy O{\textquoteright}Leary, and Eve Marder.
\newblock Ion channel degeneracy enables robust and tunable neuronal firing
  rates.
\newblock {\em Proceedings of the National Academy of Sciences},
  112(38):E5361--E5370, September 2015.

\bibitem{druckmann_novel_2007}
Shaul Druckmann, Yoav Banitt, Albert Gidon, Felix Sch{\"u}rmann, Henry Markram,
  and Idan Segev.
\newblock A {Novel} {Multiple} {Objective} {Optimization} {Framework} for
  {Constraining} {Conductance}-{Based} {Neuron} {Models} by {Experimental}
  {Data}.
\newblock {\em Frontiers in Neuroscience}, 1(1):7--18, October 2007.

\bibitem{ermentrout_mathematical_2010}
G.~Bard Ermentrout and David~H. Terman.
\newblock {\em Mathematical {Foundations} of {Neuroscience}}.
\newblock Springer, New York, 2010.

\bibitem{forssell_closed-loop_1999}
Urban Forssell and Lennart Ljung.
\newblock Closed-loop identification revisited.
\newblock {\em Automatica}, 35(7):1215--1241, July 1999.

\bibitem{geit_automated_2008}
W.~Van Geit, E.~De Schutter, and P.~Achard.
\newblock Automated neuron model optimization techniques: a review.
\newblock {\em Biological Cybernetics}, 99(4-5):241--251, November 2008.

\bibitem{gerstner_neuronal_2014}
Wulfram Gerstner, Werner~M. Kistler, Richard Naud, and Liam Paninski.
\newblock {\em Neuronal {Dynamics}: {From} {Single} {Neurons} to {Networks} and
  {Models} of {Cognition}}.
\newblock Cambridge University Press, Cambridge, UK, 2014.

\bibitem{goldwyn_what_2011}
Joshua~H. Goldwyn and Eric Shea-Brown.
\newblock The {What} and {Where} of {Adding} {Channel} {Noise} to the
  {Hodgkin}-{Huxley} {Equations}.
\newblock {\em PLOS Computational Biology}, 7(11):e1002247, November 2011.

\bibitem{hille_ionic_1984}
Bertil Hille.
\newblock {\em Ionic channels of excitable membranes}.
\newblock Sinauer Associates, Sunderland, MA, 1984.

\bibitem{hodgkin_quantitative_1952}
A.~L. Hodgkin and A.~F. Huxley.
\newblock A quantitative description of membrane current and its application to
  conduction and excitation in nerve.
\newblock {\em The Journal of Physiology}, 117(4):500--544, August 1952.

\bibitem{horn_matrix_1985}
Roger~A. Horn and Charles~R. Johnson, editors.
\newblock {\em Matrix {Analysis}}.
\newblock Cambridge University Press, Cambridge, UK, 1985.

\bibitem{huys_efficient_2006}
Quentin J.~M. Huys, Misha~B. Ahrens, and Liam Paninski.
\newblock Efficient {Estimation} of {Detailed} {Single}-{Neuron} {Models}.
\newblock {\em Journal of Neurophysiology}, 96(2):872--890, August 2006.

\bibitem{izhikevich_dynamical_2007}
Eugene~M. Izhikevich.
\newblock {\em Dynamical {Systems} in {Neuroscience}}.
\newblock MIT Press, Cambridge, MA, 2007.

\bibitem{keener_mathematical_2009}
James Keener, James Sneyd, S.~S. Antman, J.~E. Marsden, and L.~Sirovich,
  editors.
\newblock {\em Mathematical {Physiology}}, volume 8/1 of {\em Interdisciplinary
  {Applied} {Mathematics}}.
\newblock Springer, New York, NY, 2009.

\bibitem{khalil_nonlinear_2002}
Hassan~K. Khalil.
\newblock {\em Nonlinear {Systems}}.
\newblock Prentice Hall, Upper Saddle River, NJ, 3 edition, 2002.

\bibitem{koch_methods_1989}
Christof Koch and Idan Segev.
\newblock {\em Methods in {Neuronal} {Modeling}: {From} {Synapses} to
  {Networks}}.
\newblock MIT Press, Cambridge, MA, 1989.

\bibitem{lepora_efficient_2012}
Nathan~F. Lepora, Paul~G. Overton, and Kevin Gurney.
\newblock Efficient fitting of conductance-based model neurons from somatic
  current clamp.
\newblock {\em Journal of Computational Neuroscience}, 32(1):1--24, February
  2012.

\bibitem{ljung_convergence_1978}
L.~Ljung.
\newblock Convergence analysis of parametric identification methods.
\newblock {\em IEEE Transactions on Automatic Control}, 23(5):770--783, October
  1978.

\bibitem{ljung_system_1999}
Lennart Ljung.
\newblock {\em System {Identification}: {Theory} for the {User}}.
\newblock Prentice Hall PTR, Upper Saddle River, NJ, 1999.

\bibitem{ljung_perspectives_2010}
Lennart Ljung.
\newblock Perspectives on system identification.
\newblock {\em Annual Reviews in Control}, 34(1):1--12, April 2010.

\bibitem{lohmiller_contraction_1998}
Winfried Lohmiller and Jean-Jacques~E. Slotine.
\newblock On {Contraction} {Analysis} for {Non}-linear {Systems}.
\newblock {\em Automatica}, 34(6):683--696, June 1998.

\bibitem{manchester_identification_2011-1}
I.~R. Manchester, M.~M. Tobenkin, and J.~Wang.
\newblock Identification of nonlinear systems with stable oscillations.
\newblock In {\em 2011 50th {IEEE} {Conference} on {Decision} and {Control} and
  {European} {Control} {Conference}}, pages 5792--5797, Orlando, FL, December
  2011.

\bibitem{mcdougal_twenty_2017}
Robert~A. McDougal, Thomas~M. Morse, Ted Carnevale, Luis Marenco, Rixin Wang,
  Michele Migliore, Perry~L. Miller, Gordon~M. Shepherd, and Michael~L. Hines.
\newblock Twenty years of {ModelDB} and beyond: building essential modeling
  tools for the future of neuroscience.
\newblock {\em Journal of Computational Neuroscience}, 42(1):1--10, February
  2017.

\bibitem{milescu_maximum_2005}
Lorin~S. Milescu, Gustav Akk, and Frederick Sachs.
\newblock Maximum {Likelihood} {Estimation} of {Ion} {Channel} {Kinetics} from
  {Macroscopic} {Currents}.
\newblock {\em Biophysical Journal}, 88(4):2494--2515, April 2005.

\bibitem{nogaret_automatic_2016}
Alain Nogaret, C.~Daniel Meliza, Daniel Margoliash, and Henry D.~I. Abarbanel.
\newblock Automatic {Construction} of {Predictive} {Neuron} {Models} through
  {Large} {Scale} {Assimilation} of {Electrophysiological} {Data}.
\newblock {\em Scientific Reports}, 6:32749, September 2016.

\bibitem{novara_parametric_2011}
C.~Novara, T.~Vincent, K.~Hsu, M.~Milanese, and K.~Poolla.
\newblock Parametric identification of structured nonlinear systems.
\newblock {\em Automatica}, 47(4):711--721, April 2011.

\bibitem{paduart_identification_2010}
Johan Paduart, Lieve Lauwers, Jan Swevers, Kris Smolders, Johan Schoukens, and
  Rik Pintelon.
\newblock Identification of nonlinear systems using {Polynomial} {Nonlinear}
  {State} {Space} models.
\newblock {\em Automatica}, 46(4):647--656, April 2010.

\bibitem{rowat_interspike_2007}
Peter Rowat.
\newblock Interspike {Interval} {Statistics} in the {Stochastic}
  {Hodgkin}-{Huxley} {Model}: {Coexistence} of {Gamma} {Frequency} {Bursts} and
  {Highly} {Irregular} {Firing}.
\newblock {\em Neural Computation}, 19(5):1215--1250, May 2007.

\bibitem{russo_global_2010}
Giovanni Russo, Mario~di Bernardo, and Eduardo~D. Sontag.
\newblock Global {Entrainment} of {Transcriptional} {Systems} to {Periodic}
  {Inputs}.
\newblock {\em PLOS Computational Biology}, 6(4):e1000739, April 2010.

\bibitem{schon_sequential_2015}
Thomas~B. Sch{\"o}n, Fredrik Lindsten, Johan Dahlin, Johan W{\r a}gberg,
  Christian~A. Naesseth, Andreas Svensson, and Liang Dai.
\newblock Sequential {Monte} {Carlo} {Methods} for {System} {Identification}.
\newblock {\em IFAC-PapersOnLine}, 48(28):775--786, January 2015.

\bibitem{schoukens_identification_2017}
Maarten Schoukens and Koen Tiels.
\newblock Identification of block-oriented nonlinear systems starting from
  linear approximations: {A} survey.
\newblock {\em Automatica}, 85:272--292, November 2017.

\bibitem{soudry_conductance-based_2012}
Daniel Soudry and Ron Meir.
\newblock Conductance-{Based} {Neuron} {Models} and the {Slow} {Dynamics} of
  {Excitability}.
\newblock {\em Frontiers in Computational Neuroscience}, 6:4, February 2012.

\bibitem{wang_partial_2004}
Wei Wang and Jean-Jacques~E. Slotine.
\newblock On partial contraction analysis for coupled nonlinear oscillators.
\newblock {\em Biological Cybernetics}, 92(1):38--53, December 2004.

\bibitem{yuz_sampled-data_2005}
J.~I. Yuz and G.~C. Goodwin.
\newblock On sampled-data models for nonlinear systems.
\newblock {\em IEEE Transactions on Automatic Control}, 50(10):1477--1489,
  October 2005.

\end{thebibliography}

\end{document}